\documentclass[a4paper,onecolumn,11pt,accepted=2020-04-14]{quantumarticle}
\pdfoutput=1
\usepackage[utf8]{inputenc}
\usepackage[english]{babel}
\usepackage[T1]{fontenc}

\usepackage{amssymb,amsmath,amsthm,amsfonts,dsfont}
\usepackage[numbers,comma,sort&compress]{natbib}
\usepackage[colorlinks]{hyperref}
\usepackage[all]{hypcap}
\usepackage{url}
\usepackage{graphicx}
\usepackage[font=small]{caption}
\usepackage[labelformat=simple]{subcaption}
\usepackage{float}
\usepackage{algorithmic}
\usepackage[ruled,vlined,linesnumbered]{algorithm2e}
\usepackage{footnote}
\usepackage{gensymb}
\usepackage[margin=1in]{geometry}

\usepackage{xcolor}
\usepackage{pgfplots}
\definecolor{mycolor1}{rgb}{0.00000,0.44700,0.74100}%

\usepackage{mathtools}
\mathtoolsset{centercolon}
\DeclarePairedDelimiter\bra{\langle}{\rvert}
\DeclarePairedDelimiter\ket{\lvert}{\rangle}
\DeclarePairedDelimiterX\braket[2]{\langle}{\rangle}{#1 \delimsize\vert #2}

\newtheorem{theorem}{Theorem}

\newtheorem{lemma}[theorem]{Lemma}

\newtheorem{corollary}[theorem]{Corollary}

\newcommand{\eq}[1]{(\ref{eq:#1})}
\newcommand{\thm}[1]{\hyperref[thm:#1]{Theorem~\ref*{thm:#1}}}
\newcommand{\defn}[1]{\hyperref[defn:#1]{Definition~\ref*{defn:#1}}}
\newcommand{\lem}[1]{\hyperref[lem:#1]{Lemma~\ref*{lem:#1}}}
\newcommand{\prop}[1]{\hyperref[prop:#1]{Proposition~\ref*{prop:#1}}}
\newcommand{\fig}[1]{\hyperref[fig:#1]{Figure~\ref*{fig:#1}}}
\newcommand{\tab}[1]{\hyperref[tab:#1]{Table~\ref*{tab:#1}}}
\renewcommand{\sec}[1]{\hyperref[sec:#1]{Section~\ref*{sec:#1}}}
\newcommand{\append}[1]{\hyperref[append:#1]{Appendix~\ref*{append:#1}}}
\newcommand{\cor}[1]{\hyperref[cor:#1]{Corollary~\ref*{cor:#1}}}
\newcommand{\obs}[1]{\hyperref[obs:#1]{Observation~\ref*{obs:#1}}}

\newcommand{\R}{\mathbb{R}}
\newcommand{\C}{\mathbb{C}}

\newcommand{\comment}[1]{}
\newcommand{\norm}[1]{\left\lVert#1\right\rVert}

\pgfplotsset{
	log x ticks with fixed point/.style={
		xticklabel={
			\pgfkeys{/pgf/fpu=true}
			\pgfmathparse{exp(\tick)}%
			\pgfmathprintnumber[fixed relative, precision=3]{\pgfmathresult}
			\pgfkeys{/pgf/fpu=false}
		}
	},
	log y ticks with fixed point/.style={
		yticklabel={
			\pgfkeys{/pgf/fpu=true}
			\pgfmathparse{exp(\tick)}%
			\pgfmathprintnumber[fixed relative, precision=3]{\pgfmathresult}
			\pgfkeys{/pgf/fpu=false}
		}
	}
}

\newenvironment{proof-sketch}{%
	\proof}{\endproof}

\usepackage{thmtools}
\usepackage{thm-restate}

\usepackage{ mathrsfs }

\newcommand{\E}{\mathrm{E}}
\newcommand{\CE}{\mathcal{E}}
\newcommand{\CU}{\mathcal{U}}
\newcommand{\CG}{\mathcal{G}}

\newcommand{\CO}{\mathcal{O}}

\makeatletter
\newcommand{\neutralize}[1]{\expandafter\let\csname c@#1\endcsname\count@}
\makeatother

\newenvironment{thmbis}[1]
{\renewcommand{\thetheorem}{\ref*{#1}$'$}%
	\neutralize{thm}\phantomsection
	\begin{theorem}}
	{\end{theorem}
	\addtocounter{theorem}{-1}%
}

\newenvironment{corbis}[1]
{\renewcommand{\thetheorem}{\ref*{#1}$'$}%
	\neutralize{thm}\phantomsection
	\begin{corollary}}
	{\end{corollary}
	\addtocounter{theorem}{-1}%
}

\usepackage{enumerate}

\allowdisplaybreaks

\title{\protect\parbox{\textwidth}{\protect\centering Time-dependent Hamiltonian simulation\\ with $L^1$-norm scaling}}

\author[aff1]{Dominic W.\ Berry}
\orcid{0000-0003-3446-1449}
\author[aff2,aff3]{Andrew M.\ Childs}
\orcid{0000-0002-9903-837X}
\author[aff2,aff3]{Yuan Su}
\orcid{0000-0003-1144-3563}
\author[aff3,aff4]{Xin Wang}
\orcid{0000-0002-0641-3186}
\author[aff5,aff6,aff7]{Nathan Wiebe}
\orcid{0000-0001-6047-0547}
\affiliation[aff1]{Department of Physics and Astronomy, Macquarie University, Sydney, NSW 2109, Australia}
\affiliation[aff2]{Department of Computer Science, University of Maryland, College Park, MD 20742, USA}
\affiliation[aff3]{Institute for Advanced Computer Studies and Joint Center for Quantum Information and Computer Science, University of Maryland, College Park, MD 20742, USA}
\affiliation[aff4]{Institute for Quantum Computing, Baidu Research, Beijing 100193, China}
\affiliation[aff5]{Department of Physics, University of Washington, Seattle, WA 98195, USA}
\affiliation[aff6]{Pacific Northwest National Laboratory, Richland, WA 99354, USA}
\affiliation[aff7]{Google Inc., Venice, CA 90291, USA}

\begin{document}
\maketitle

\hypersetup{pdfauthor={Dominic W. Berry, Andrew M. Childs, Yuan Su, Xin Wang, Nathan Wiebe},
	pdftitle={Time-dependent Hamiltonian simulation with L1-norm scaling},
}
	
\begin{abstract}
	The difficulty of simulating quantum dynamics depends on the norm of the Hamiltonian. When the Hamiltonian varies with time, the simulation complexity should only depend on this quantity instantaneously. We develop quantum simulation algorithms that exploit this intuition. For sparse Hamiltonian simulation, the gate complexity scales with the $L^1$ norm $\int_{0}^{t}\mathrm{d}\tau\norm{H(\tau)}_{\max}$, whereas the best previous results scale with $t\max_{\tau\in[0,t]}\norm{H(\tau)}_{\max}$. We also show analogous results for Hamiltonians that are linear combinations of unitaries. Our approaches thus provide an improvement over previous simulation algorithms that can be substantial when the Hamiltonian varies significantly. We introduce two new techniques: a classical sampler of time-dependent Hamiltonians and a rescaling principle for the Schr\"{o}dinger equation. The rescaled Dyson-series algorithm is nearly optimal with respect to all parameters of interest, whereas the sampling-based approach is easier to realize for near-term simulation. These algorithms could potentially be applied to semi-classical simulations of scattering processes in quantum chemistry.
\end{abstract}

\newpage
{
	\clearpage\tableofcontents
}
\newpage

\section{Introduction}
\label{sec:intro}

Simulating the Hamiltonian dynamics of a quantum system is one of the most promising applications of a quantum computer. The apparent classical intractability of simulating quantum dynamics led Feynman \cite{Fey82} and others to propose the idea of quantum computation. Quantum computers can simulate various physical systems, including condensed matter physics \cite{BWMMNC18}, quantum field theory \cite{JLP12}, and quantum chemistry \cite{BGBWMPFN18,Cao18,WBA11,MEABY20}. The study of quantum simulation has also led to the discovery of new quantum algorithms, such as algorithms for linear systems \cite{HHL09}, differential equations \cite{BCOW17}, semidefinite optimization \cite{BS17}, formula evaluation \cite{FGG07}, quantum walk \cite{CCDFGS03}, and ground-state and thermal-state preparation \cite{PW09,CS16}.

Let $H(\tau)$ be a Hamiltonian defined for $0\leq\tau\leq t$. The problem of Hamiltonian simulation is to approximate the evolution $\exp_{\mathcal{T}}\big({-}i\int_{0}^{t}\mathrm{d}\tau\ H(\tau)\big)$ using a quantum circuit comprised of elementary quantum gates, where $\exp_{\mathcal{T}}$ denotes the time-ordered matrix exponential. If the Hamiltonian $H(\tau)=H$ does not depend on time, the evolution operator can be represented in closed form as $e^{-itH}$. Then the problem can be greatly simplified and it has been thoroughly studied by previous works on quantum simulation \cite{Llo96,AT03,BACS05,BCK15,BC12,BCCKS14,LC16,LC17,Low18,BCG14,Campbell18,Kieferova18,LW18,CMNRS18,COS18,CS19}.

Simulating a general time-dependent Hamiltonian $H(\tau)$ naturally subsumes the time-independent case, and can be applied to devising quantum control schemes \cite{Pang2017,Nielsen06}, describing quantum chemical reactions \cite{Butler98}, and implementing adiabatic quantum algorithms \cite{Farhi01}. However, the problem becomes considerably harder and there are fewer quantum algorithms available. Wiebe, Berry, H\o{}yer, and Sanders designed a time-dependent Hamiltonian simulation algorithm based on higher-order product formulas \cite{WBHS10}. They assume that $H(\tau)$ is smooth up to a certain order and they give an example in which a desired approximation cannot be achieved due to the non-differentiability of the Hamiltonian. The smoothness assumption is relaxed in subsequent work by Poulin, Qarry, Somma, and Verstraete \cite{PQSV11} based on techniques of Hamiltonian averaging and Monte Carlo estimation. The fractional-query algorithm of Berry, Childs, Cleve, Kothari, and Somma can also simulate time-dependent Hamiltonians \cite{FractionalQuery14}, with an exponentially improved dependence on precision and only logarithmic dependence on the derivative of the Hamiltonian. A related quantum algorithm for time-dependent Hamiltonian simulation was suggested by Berry, Childs, Cleve, Kothari, and Somma based on the truncated Dyson series \cite{BCCKS14}, which is analyzed explicitly in \cite{Kieferova18,LW18}.

In this paper, we study time-dependent Hamiltonian simulation based on a simple intuition: the difficulty of simulating a quantum system should depend on the integrated norm of the Hamiltonian. To elaborate, first consider the special case of simulating a time-independent Hamiltonian. The complexity of such a simulation depends on $t\norm{H}$ \cite{CK10}, where $\norm{\cdot}$ is a matrix norm that quantifies the size of the Hamiltonian.
It is common to express the complexity in terms of the spectral norm $\norm{H}_\infty$ (i.e., the Schatten $\infty$-norm), which quantifies the maximum energy of $H$.

In the general case where the Hamiltonian $H(\tau)$ is time dependent, 
we expect a quantum simulation algorithm to depend on the Hamiltonian locally in time, and therefore to have complexity that scales with the integrated norm $\int_{0}^{t}\mathrm{d}\tau\norm{H(\tau)}$. This is the $L^1$ norm of $\norm{H(\tau)}$ when viewed as a function of $\tau$, so we say such an algorithm has $L^1$-norm scaling. Surprisingly, existing analyses of quantum simulation algorithms fail to achieve this complexity; rather, their gate complexity scales with the worst-case cost $t\max_{\tau\in[0,t]}\norm{H(\tau)}$. It is therefore reasonable to question whether our intuition is correct, or if there exist faster time-dependent Hamiltonian simulation algorithms that can exploit this intuition.\footnotemark

\footnotetext{For the Dyson-series approach, Low and Wiebe claimed that the worst-case scaling may be avoided by a proper segmentation of the time interval \cite{LW18}. However, it is unclear how their analysis can be formalized to give an algorithm with complexity that scales with the $L^1$ norm. In \sec{rescaling}, we propose a rescaling principle for the Schr\"{o}dinger equation and develop a rescaled Dyson-series algorithm with $L^1$-norm scaling.}

Our work answers this question by providing multiple faster quantum algorithms for time-dependent Hamiltonian simulation. These algorithms have gate complexity that scales with $\int_{0}^{t}\mathrm{d}\tau\norm{H(\tau)}$, in contrast to the best previous scaling of $t\max_{\tau\in[0,t]}\norm{H(\tau)}$. As the norm inequality $\int_{0}^{t}\mathrm{d}\tau\norm{H(\tau)}\leq t\max_{\tau\in[0,t]}\norm{H(\tau)}$ always holds but is not saturated in general, these algorithms provide a strict speedup over existing algorithms. We further analyze an application to simulating scattering processes in quantum chemistry, showing that our improvement can be favorable in practice.

We introduce notation and terminology and state our assumptions in \sec{prelim}.
Following standard assumptions about quantum simulation, we consider two different input models of Hamiltonians. The first is the sparse matrix (SM) model common for analyzing Hamiltonian simulation in general, in which the Hamiltonian is assumed to be sparse and access to the locations and values of nonzero matrix elements are provided by oracles. We quantify the complexity of a simulation algorithm by the number of queries and additional gates it uses. The second model, favorable for practical applications such as condensed matter physics and quantum chemistry simulation, assumes that the Hamiltonian can be explicitly decomposed as a linear combination of unitaries (LCU), where the coefficients are efficiently computable on a classical computer and the summands can be exponentiated and controlled on a quantum computer. We ignore the cost of implementing the coefficient oracle and focus mainly on the gate complexity. Quantum simulation algorithms can sometimes work more efficiently in other input models, but we study these two models since they are versatile and provide a fair comparison of the gate complexity.

Reference \cite{FractionalQuery14} claims that the fractional-query algorithm can simulate time-dependent Hamiltonians with $L^\infty$-norm scaling. However, it is not hard to see that its query complexity in fact scales with the $L^1$ norm. While we do not show how to achieve this scaling for the gate complexity, our analysis is simple and suggests that such a result might be possible. We analyze the query complexity of the fractional-query algorithm in \sec{prelim_fracquery}.

We develop two new techniques to simulate time-dependent Hamiltonians with $L^1$-norm scaling. Our first technique is a classical sampling protocol for time-dependent Hamiltonians. In this protocol, we randomly sample a time $\tau\in[0,t]$ and evolve under the time-independent Hamiltonian $H(\tau)$, where the probability distribution is designed to favor those $\tau$ with large $\norm{H(\tau)}$. Campbell introduced a discrete sampling scheme for time-independent Hamiltonian simulation \cite{Campbell18} and our protocol can be viewed as its continuous analog, which we call ``continuous qDRIFT''. We show that continuous qDRIFT is universal, in the sense that any Hamiltonian simulable by \cite{Campbell18} can be simulated by continuous qDRIFT with the same complexity. In addition, we shave off a multiplicative factor in the analysis of \cite{Campbell18} by explicitly evaluating the rate of change of the evolution with respect to scaling the Hamiltonian. Continuous qDRIFT and its analysis are detailed in \sec{qdrift}. Our algorithm is also similar in spirit to the approach of Poulin et al.\ \cite{Pou15} based on Hamiltonian averaging and Monte Carlo estimation, although their algorithm does not have $L^1$-norm scaling. We discuss the relationship between these two approaches in \append{poulin}.

We also present a general principle for rescaling the Schr\"{o}dinger equation in \sec{rescaling}. In the rescaled Schr\"{o}dinger equation, the time-dependent Hamiltonian $H(\tau)$ has the same norm for all $\tau\in[0,t]$, so the norm inequality $\int_{0}^{t}\mathrm{d}\tau\norm{H(\tau)}\leq t\max_{\tau\in[0,t]}\norm{H(\tau)}$ holds with equality. 
Using this principle, we show that the simulation algorithm based on the truncated Dyson series \cite{BCCKS14,LW18,Kieferova18} can also be improved to have $L^1$-norm scaling.

To illustrate how our results might be applied, we identify a specific problem in quantum chemistry for which our $L^1$-norm improvement is advantageous: semi-classical scattering of molecules in a chemical reaction. For such a simulation, $\norm{H(\tau)}$ changes dramatically throughout the evolution, so its $L^1$ norm can be significantly smaller than its $L^\infty$ norm. We discuss this application further in \sec{scatter}.

Finally, we conclude in \sec{discussion} with a brief discussion of the results and some open questions.

\section{Preliminaries}
\label{sec:prelim}

\subsection{Time-dependent Hamiltonian evolution}
\label{sec:prelim_evolve}
Let $H(\tau)$ be a time-dependent Hamiltonian defined for $0\leq\tau\leq t$. By default, we assume that $H(\tau)$ is continuously differentiable and $H(\tau)\neq 0$ everywhere, and we defer the discussion of pathological cases to \sec{discussion}. If the Hamiltonian $H(\tau)=H$ is time independent, the evolution is given in closed form by the matrix exponential $e^{-itH}$. However, there exists no such closed-form expression for a general $H(\tau)$ and we instead represent the evolution by $\exp_{\mathcal{T}}\big(-i\int_{0}^{t}\mathrm{d}\tau\, H(\tau)\big)$, where $\exp_{\mathcal{T}}$ denotes the time-ordered matrix exponential. We have
\begin{equation}
	\frac{\mathrm{d}}{\mathrm{d}t}\exp_{\mathcal{T}}\bigg(-i\int_{0}^{t}\mathrm{d}\tau\, H(\tau)\bigg)
	=-iH(t)\exp_{\mathcal{T}}\bigg(-i\int_{0}^{t}\mathrm{d}\tau\, H(\tau)\bigg).
\end{equation}
If $G(\tau)$ is another time-dependent Hamiltonian, the evolutions generated by $H(\tau)$ and $G(\tau)$ have distance bounded by the following lemma.

\begin{lemma}[$L^1$-norm distance bound of time-ordered evolutions {\cite[Appendix B]{Tran18}}]
	\label{lem:generator}
	Let $H(\tau)$ and $G(\tau)$ be time-dependent Hamiltonians defined on the interval $0\leq\tau\leq t$. Then,
	\begin{equation}
	\norm{\exp_{\mathcal{T}}\bigg(-i\int_{0}^{t}\mathrm{d}\tau\, H(\tau)\bigg)
		-\exp_{\mathcal{T}}\bigg(-i\int_{0}^{t}\mathrm{d}\tau\, G(\tau)\bigg)}_\infty
	\leq\int_{0}^{t}\mathrm{d}\tau\norm{H(\tau)-G(\tau)}_\infty.
	\end{equation}
	Here, $\norm{\cdot}_\infty$ denotes the spectral norm.
\end{lemma}

We will abbreviate the evolution operator as $\E(t,s):=\exp_{\mathcal{T}}\big({-}i\int_{s}^{t}\mathrm{d}\tau\, H(\tau)\big)$ when there is no ambiguity. In the special case where $H(\tau)=H$ is time independent, the evolution $e^{-itH}$ only depends on the time duration so we denote $\E(t):=\E(t,0)$. Therefore, we have the differential equation
\begin{equation}
	\frac{\mathrm{d}}{\mathrm{d}t}\E(t,0)=-iH(t)\E(t,0),\qquad
	\E(0,0)=I.
\end{equation}
We may further obtain an integral representation of $\E(t,0)$. To this end, we apply the fundamental theorem of calculus to the Schr\"{o}dinger equation and obtain
\begin{equation}
	\E(t,0)-I=\E(t,0)-\E(0,0)=\int_{0}^{t}\mathrm{d}\tau\, \frac{\mathrm{d}}{\mathrm{d}\tau}\E(\tau,0)
	=-i\int_{0}^{t}\mathrm{d}\tau\, H(\tau)\E(\tau,0).
\end{equation}
Equivalently, $\E(t,0)$ satisfies the integral equation
\begin{equation}
	\E(t,0)=I-i\int_{0}^{t}\mathrm{d}\tau\, H(\tau)\E(\tau,0).
\end{equation}

For any $0\leq s\leq t$, the evolution operator satisfies the multiplicative property
\begin{equation}
	\E(t,0)=\E(t,s)\E(s,0).
\end{equation}
The operator $\E(0,t)$ with $t\geq 0$ is understood as the inverse evolution operator
\begin{equation}
	\E(0,t):=\E^{-1}(t,0)=\E^{\dagger}(t,0).
\end{equation}
For a thorough mathematical treatment of time-dependent Hamiltonian evolution, we refer the reader to \cite{bk:dollard_friedman}. Finally, the quantum channel corresponding to the unitary evolution $\E(t,0)$ is denoted as $\CE(t,0)$ and is defined by
\begin{equation}
	\CE(t,0)(\rho):=\E(t,0)\rho\E^\dagger(t,0)=\E(t,0)\rho\E(0,t).
\end{equation}
For time-independent Hamiltonians, we denote $\CE(t):=\CE(t,0)$.

\subsection{Notation for norms}
\label{sec:prelim_norm}
We introduce norm notation for vectors, matrices, operator-valued functions, and linear maps on the space of matrices.

Let $\alpha=\begin{bmatrix}\alpha_1 & \alpha_2 & \cdots & \alpha_L\end{bmatrix}\in\C^L$ be an $L$-dimensional vector. We use $\norm{\alpha}_{p}$ to represent the vector $\ell_p$ norm of $\alpha$. Thus,
\begin{equation}
	\norm{\alpha}_1:=\sum_{j=1}^{L}|\alpha_j|,\qquad
	\norm{\alpha}_2:=\sqrt{\sum_{j=1}^{L}|\alpha_j|^2},\qquad
	\norm{\alpha}_\infty:=\max_{j\in\{1,2,\ldots,L\}}|\alpha_j|.
\end{equation}
For a matrix $A$, we define $\norm{A}_p$ to be the Schatten $p$-norm of $A$ \cite{bk:wat,bk:wilde}. We have
\begin{equation}
	\norm{A}_1:=\mathrm{Tr}\big(\sqrt{A^\dagger A}\big),\qquad
	\norm{A}_2:=\sqrt{\mathrm{Tr}\big(A^\dagger A\big)},\qquad
	\norm{A}_\infty:=\max_{|\psi\rangle}\norm{A|\psi\rangle}_2.
\end{equation}
Finally, if $f\colon[0,t]\rightarrow\C$ is a continuous function, we use $\norm{f}_p$ to mean the $L^p$ norm of the function $f$. Thus,
\begin{equation}
	\norm{f}_1:=\int_{0}^{t}\mathrm{d}\tau\, |f(\tau)|,\qquad
	\norm{f}_2:=\sqrt{\int_{0}^{t}\mathrm{d}\tau\, |f(\tau)|^2},\qquad
	\norm{f}_\infty:=\max_{\tau\in[0,t]}|f(\tau)|.
\end{equation}

We combine these norms to obtain norms for vector-valued and operator-valued functions. Let $\alpha\colon [0,t]\rightarrow\C^L$ be a continuous vector-valued function, with the $j$th coordinate at time $\tau$ denoted $\alpha_j(\tau)$. We use $\norm{\alpha}_{p,q}$ to mean that we take the $\ell_p$ norm $\norm{\alpha(\tau)}_p$ for every $\tau$ and compute the $L^q$ norm of the resulting scalar function. For example,
\begin{equation}
	\norm{\alpha}_{1,1}:=\int_{0}^{t}\mathrm{d}\tau\sum_{j=1}^{L}|\alpha_j(\tau)|,\qquad
	\norm{\alpha}_{1,\infty}:=\max_{\tau\in[0,t]}\sum_{j=1}^{L}|\alpha_j(\tau)|.
\end{equation}
Note that $\norm{\alpha(\tau)}_p$ is continuous as a function of $\tau$, so $\norm{\alpha}_{p,q}$ is well defined and is indeed a norm for vector-valued functions. Similarly, we also define $\norm{A}_{p,q}$ for a continuous operator-valued function by taking the Schatten $p$-norm $\norm{A(\tau)}_p$ for every $\tau$ and computing the $L^q$ norm of the resulting scalar function. In rare cases, we will also encounter time-dependent linear combinations of operators of the form $A(\tau)=\sum_{l=1}^{L}A_l(\tau)$, and we write $\norm{A}_{p,q,r}$ to mean that we take the Schatten $p$-norm $\norm{A_l(\tau)}_p$ of each summand, and apply the $\ell_q$ norm and $L^r$ norm to the resulting vector-valued functions. For example,
\begin{equation}
	\norm{A}_{2,1,\infty}:=\max_{\tau\in[0,t]}\sum_{l=1}^{L}\norm{A_l(\tau)}_2.
\end{equation}

We also define $\norm{A}_{\max}$ as the largest matrix element of $A$ in absolute value,
\begin{equation}
	\norm{A}_{\max}:=\max_{j,k}|A_{j,k}|.
\end{equation}
The norm $\norm{A}_{\max}$ is a vector norm of $A$ but does not satisfy the submultiplicative property of a matrix norm. It relates to the spectral norm by the inequality \cite[Lemma 1]{CK10}
\begin{equation}
	\norm{A}_{\max}\leq\norm{A}_\infty.
\end{equation}
If $A$ is a continuous operator-valued function, we interpret $\norm{A}_{\max,q}$ in a similar way as above. Therefore,
\begin{equation}
	\norm{A}_{\max,1}:=\int_{0}^{t}\mathrm{d}\tau\norm{A(\tau)}_{\max},\qquad
	\norm{A}_{\max,\infty}:=\max_{\tau\in[0,t]}\norm{A(\tau)}_{\max}.
\end{equation}

Finally, we define a norm for linear maps on the space of matrices. Let $\mathcal{E}\colon A\mapsto \mathcal{E}(A)$ be a linear map on the space of matrices on $\mathcal{H}$. The diamond norm of $\mathcal{E}$ is
\begin{equation}
\norm{\mathcal{E}}_\diamond:=\max\{\norm{(\mathcal{E}\otimes\mathds{1}_\mathcal{H})(B)}_1:\norm{B}_1\leq 1\},
\end{equation}
where the maximization is taken over all matrices $B$ on $\mathcal{H}\otimes\mathcal{H}$ satisfying $\norm{B}_1\leq 1$. Below is a useful bound on the diamond-norm distance between two unitary channels.

\begin{lemma}[Diamond-norm distance between unitary channels {\cite[Lemma 7]{BCK15}}]
	\label{lem:diamond_bound}
	Let $V$ and $U$ be unitary matrices, with associated quantum channels $\mathcal{V}:\rho\mapsto V\rho V^\dagger$ and $\mathcal{U}:\rho\mapsto U\rho U^\dagger$. Then,
	\begin{equation}
		\norm{\mathcal{U}-\mathcal{V}}_\diamond \leq 2\norm{U-V}_\infty.
	\end{equation}
\end{lemma}

The sampling-based algorithm (\sec{qdrift}) produces a channel close to $\CE(t,0)(\rho)=\exp_{\mathcal{T}}\big(-i\int_{0}^{t}\mathrm{d}\tau H(\tau)\big)\rho\exp_{\mathcal{T}}^\dagger\big(-i\int_{0}^{t}\mathrm{d}\tau H(\tau)\big)$, and its error is naturally quantified by the diamond-norm distance. Other simulation algorithms such as the Dyson-series approach (\sec{rescaling}) produce operators that are close to the unitary $\exp_{\mathcal{T}}\big(-i\int_{0}^{t}\mathrm{d}\tau H(\tau)\big)$, and we quantify their error in terms of the spectral norm. For a fair comparison one may instead describe all simulation algorithms using quantum channels and use the diamond-norm distance as the unified error metric. By \lem{diamond_bound}, we lose at most a factor of $2$ in this conversion.

\subsection{Hamiltonian input models}
\label{sec:prelim_model}
Quantum simulation algorithms may have different performance depending on the choice of the input model of Hamiltonians. In this section, we describe two input models that are extensively used in previous works: the sparse matrix (SM) model and the linear-combination-of-unitaries (LCU) model. We also discuss other input models that will be used in later sections.

Let $H(\tau)$ be a time-dependent Hamiltonian defined for $0\leq\tau\leq t$. In the SM model, we assume that $H(\tau)$ is $d$-sparse in the sense that the number of nonzero matrix elements within each row and column throughout the entire interval $[0,t]$ is at most $d$. We assume that the locations of the nonzero matrix elements are time independent. Access to the Hamiltonian is given through the oracles
\begin{equation}
\begin{aligned}
\CO_{\text{loc}}|j,s\rangle&=|j,\mathrm{col}(j,s)\rangle,\\
\CO_{\text{val}}|\tau,j,k,z\rangle&=|\tau,j,k,z\oplus H_{jk}(\tau)\rangle.
\end{aligned}
\end{equation}
Here, $\mathrm{col}(j,s)$ returns the column index of the $s$th element in the $j$th row that may be nonzero over the entire time interval $[0,t]$. We quantify the complexity of a quantum simulation algorithm by the number of queries it makes to $\CO_{\text{loc}}$ and $\CO_{\text{val}}$, together with the number of additional elementary gates it requires. Such a model includes many realistic physical systems and is well-motivated from a theoretical perspective \cite{HHL09}.

As the following lemma shows, a $d$-sparse Hamiltonian can be efficiently decomposed as a sum of $1$-sparse terms.
\begin{lemma}[Decomposition of sparse Hamiltonians {\cite[Lemma 4.3 and 4.4]{FractionalQuery14}}]
	\label{lem:sparse_decomp}
	Let $H$ be a time-independent $d$-sparse Hamiltonian accessed through the oracles $\CO_{\text{loc}}$ and $\CO_{\text{val}}$. Then
	\begin{enumerate}
		\item there exists a decomposition $H=\sum_{j=1}^{d^2}H_j$, where each $H_j$ is $1$-sparse with $\norm{H_j}_{\max}\leq\norm{H}_{\max}$, and a query to any $H_j$ can be simulated with $O(1)$ queries to $H$; and
		\item for any $\gamma>0$, there exists an approximate decomposition\footnote{Reference \cite{FractionalQuery14} uses \cite[Lemma 4.3]{FractionalQuery14} and the triangle inequality to show that $\bigl\|{H-\gamma\sum_{j=1}^{\eta}G_j}\bigr\|_{\max}\leq\sqrt{2}\gamma d^2$. However, this bound can be tightened to $\sqrt{2}\gamma$, since the max-norm distance depends on the largest error from rounding off the $d^2$ $1$-sparse matrices. } $\bigl\|{H-\gamma\sum_{j=1}^{\eta}G_j}\bigr\|_{\max}\leq\sqrt{2}\gamma$, where $\eta=O\big(d^2\norm{H}_{\max}/\gamma\big)$, each $G_j$ is $1$-sparse with eigenvalues $\pm 1$, and a query to any $G_j$ can be simulated with $O(1)$ queries to $H$.
	\end{enumerate}
\end{lemma}

For the LCU model, we suppose that the Hamiltonian $H(\tau)$ admits a decomposition
\begin{equation}
	H(\tau)=\sum_{l=1}^{L}\alpha_l(\tau)H_l,
\end{equation}
where the coefficients $\alpha_l(\tau)\geq 0$ are continuously differentiable and nonzero everywhere, and the matrices $H_l$ are both unitary and Hermitian. We assume that the coefficients $\alpha_l(\tau)$ can be efficiently computed by a classical oracle $\CO_{\text{coeff}}$, and we ignore the classical cost of implementing such an oracle. We further assume that each $|0\rangle\langle 0|\otimes I+|1\rangle\langle 1|\otimes H_l$ can be implemented with gate complexity $g_c$, and each $|0\rangle\langle 0|\otimes I+|1\rangle\langle1|\otimes e^{-i\tau H_l}$ for an arbitrarily large $\tau$ can be performed with $g_e$ gates. Such a setting is common in the simulation of condensed matter physics and quantum chemistry. We quantify the complexity of a simulation algorithm by the number of elementary gates it uses.

Quantum simulation algorithms can sometimes work in other input models. For example, the qDRIFT protocol introduced in \sec{qdrift} requires only that the Hamiltonians have the form
\begin{equation}
	H(\tau)=\sum_{l=1}^{L}H_l(\tau),
\end{equation}
where the Hermitian-valued functions $H_l(\tau)$ are continuous, nonzero everywhere, and can be efficiently exponentiated on a quantum computer. We call this the linear combination (LC) model. On the other hand, the Dyson-series algorithm can be described in terms of the $\textsc{Select}$ operation
\begin{equation}
	\textsc{Select}(H):=\sum_{l=1}^{L}|l\rangle\langle l|\otimes H_l,
\end{equation}
irrespective of how this operation is implemented. We consider the SM and LCU models for all the time-dependent simulation algorithms so that we can give a fair comparison of their complexity. 

\subsection{Simulation algorithms with \texorpdfstring{$L^1$}{L1}-norm scaling}
\label{sec:prelim_l1}
We now explain the meaning of $L^1$-norm scaling in the SM and the LCU models.
Let $H(\tau)$ be a time-dependent Hamiltonian defined for $0\leq\tau\leq t$. We say that an algorithm in the SM model simulates $H(\tau)$ with $L^1$-norm scaling if, given any continuously differentiable upper bound $\Lambda_{\max}(\tau)\geq\norm{H(\tau)}_{\max}$, the algorithm has query complexity and gate complexity that scale with $\norm{\Lambda_{\max}}_1=\int_{0}^{t}\mathrm{d}\tau\,\Lambda_{\max}(\tau)$ up to logarithmic factors. The norm bound $\Lambda_{\max}(\tau)$, together with other auxiliary information, must be accessed by the quantum simulation algorithm; we assume such quantities can be computed efficiently.

In the LCU model, we are given a time-dependent Hamiltonian with the decomposition $H(\tau)=\sum_{l=1}^{L}\alpha_l(\tau)H_l$. We say that an algorithm has $L^1$-norm scaling if, for any continuously differentiable vector-valued function $\Lambda(t)$ with $\Lambda_l(\tau)\geq \alpha_l(\tau)$, the algorithm has query and gate complexity that scale with $\norm{\Lambda}_{\infty,1}=\int_{0}^{t}\mathrm{d}\tau\max_l\Lambda_l(\tau)$ up to logarithmic factors.

For better readability, we express the complexity of simulation algorithms in terms of the norm of the original Hamiltonian, such as $\norm{H}_{\max,1}$ and $\norm{\alpha}_{\infty,1}$, instead of the upper bounds $\norm{\Lambda_{\max}}_1$ and $\norm{\Lambda}_{\infty,1}$. 
We use standard asymptotic notation, with $O$, $\Omega$, and $\Theta$ representing asymptotic upper, lower, and tight bounds, respectively.
We also suppress logarithmic factors using the $\widetilde O$ notation when the complexity expression becomes too complicated. \tab{complexity_compare} compares the results of this paper with 
previous results on simulating time-dependent Hamiltonians.

\begin{table}
	\centering
	\renewcommand{\arraystretch}{1.4}
	\begin{tabular}{c||c|c}
		Algorithms & SM & LCU\\
		\hline\hline
		Monte Carlo estimation (first step) \cite{PQSV11} & $\widetilde{O}\big((d^2\norm{H}_{\max,\infty}t)^2n/\epsilon\big)$ & $O\big((\norm{\alpha}_{1,\infty}t)^2g_e/\epsilon\big)$\\
		Fractional-query \cite{FractionalQuery14}  & $\widetilde{O}\big(d^2\norm{H}_{\max,\infty}tn\big)$
		 &  N/A \\
		Dyson series \cite{BCCKS14,LW18,Kieferova18} & $\widetilde{O}\big(d\norm{H}_{\max,\infty}tn\big)$ & $\widetilde{O}\big(\norm{\alpha}_{\infty,\infty}tL^2g_c\big)$\\
		\hline
		Continuous qDRIFT (\sec{qdrift_complexity}) & $\widetilde{O}\big((d^2\norm{H}_{\max,1})^2n/\epsilon\big)$ & $O\big(\norm{\alpha}_{1,1}^2g_e/\epsilon\big)$\\
		Rescaled Dyson series (\sec{rescaling_dscomplexity}) & $\widetilde{O}\big(d\norm{H}_{\max,1}n\big)$ & $\widetilde{O}\big(\norm{\alpha}_{\infty,1}L^2g_c\big)$
	\end{tabular}
	\caption{Complexity comparison of previous algorithms (top three) 
	and the algorithms introduced in this paper (bottom two) for simulating time-dependent Hamiltonians. Logarithmic factors are suppressed by $\widetilde{O}$ notation and the (non-query) gate complexities are compared. The product formula algorithm of \cite{WBHS10} is omitted as its gate complexity scales polynomially with high-order derivatives and is not directly comparable to other algorithms in the table. The complexity of the full Monte Carlo estimation algorithm \cite{PQSV11} is not analyzed explicitly; only its first step is compared.
	The fractional-query algorithm \cite{FractionalQuery14} does not have an explicit implementation for Hamiltonians in the LCU model, 
	and its implementation in the SM model is streamlined by the Dyson-series approach \cite{BCCKS14,LW18,Kieferova18}.}
	\label{tab:complexity_compare}
\end{table}

Our goal is to develop simulation algorithms that scale with the $L^1$-norm with respect to the time variable $\tau$, for both query complexity and gate complexity. We start by reexamining the fractional-query approach. It was mentioned in \cite{FractionalQuery14} that this approach can simulate time-dependent Hamiltonians with $L^\infty$-norm scaling, but we find that its query complexity scales with the $L^1$ norm. We give this improved analysis in the next section.

\subsection{Query complexity with \texorpdfstring{$L^1$}{L1}-norm scaling}
\label{sec:prelim_fracquery}
We begin by reviewing the result of \cite{FractionalQuery14} for simulating time-independent Hamiltonians. We assume that the Hamiltonian is given by a linear combination of unitaries $G=\sum_{l=1}^{L}\beta_l G_l$ with nonnegative coefficients $\beta_l$. Here, $G_l$ are both unitary and Hermitian, so they are reflections and their eigenvalues are $\pm 1$.

We say that a quantum operation is a fractional-query algorithm if it is of the form
\begin{equation}
U_m Q^{\tau_m} U_{m-1}\cdots U_1 Q^{\tau_1}U_0,
\end{equation}
where $Q$ is unitary with eigenvalues $\pm 1$, $U_j$ are unitary operations, and $\tau_j\geq 0$. Here, we regard $Q$ as the oracle and $U_j$ as non-query operations, so this algorithm has fractional-query complexity $\sum_{j=1}^{m}\tau_j$. A quantum algorithm that makes (discrete) queries to $Q$ is a fractional-query algorithm with $\tau_j=1$. Conversely, any fractional-query algorithm can be efficiently simulated in the discrete query model. In particular, an algorithm with fractional-query complexity $T$ can be simulated with error at most $\epsilon$ using $O\big(T\frac{\log(T/\epsilon)}{\log\log(T/\epsilon)}\big)$ discrete queries \cite[Lemma 3.8]{FractionalQuery14}.

To apply the fractional-query approach, we approximate the evolution under $G$ using the first-order product formula
\begin{equation}
\label{eq:fracquery_timeindep_pf1}
\norm{e^{-itG}-\bigg(e^{-i\frac{t}{r} \beta_1G_1}\cdots e^{-i\frac{t}{r} \beta_LG_{L}}\bigg)^r}_\infty
=O\bigg(\frac{(\norm{\beta}_1 t)^2}{r}\bigg).
\end{equation}
Observe that $e^{-i\pi G_l}$ are unitary operations with eigenvalues $\pm 1$, so $\big(e^{-i\frac{t}{r} \beta_1G_1}\cdots e^{-i\frac{t}{r} \beta_LG_{L}}\big)^r$ can be viewed as a fractional-query algorithm with query complexity $O(\norm{\beta}_1 t)$, provided that we can make fractional queries to multiple oracles $e^{-i\pi G_1},\ldots,e^{-i\pi G_{L}}$. This can be realized by a standard fractional-query algorithm accessing the single oracle
\begin{equation}
\textsc{Select}(\textsc{Exp-}G)=\sum_{l=1}^{L}|l\rangle\langle l|\otimes e^{-i\pi G_l}
\end{equation}
with the same query complexity \cite[Theorem 4.1]{FractionalQuery14}.

To simulate with accuracy $\epsilon$, we set $r=O\big((\norm{\beta}_1 t)^2/\epsilon\big)$ to ensure that
\begin{equation}
\norm{e^{-itG}-\bigg(e^{-i\frac{t}{r} \beta_1G_1}\cdots e^{-i\frac{t}{r} \beta_LG_{L}}\bigg)^r}_\infty
=O(\epsilon).
\end{equation}
We now convert this multi-oracle algorithm to a single-oracle algorithm with the same fractional-query complexity $T=O(\norm{\beta}_1t)$ and, with precision $O(\epsilon)$, implement it in the discrete query model. Altogether, this approach makes
\begin{equation}
O\bigg(T\frac{\log(T/\epsilon)}{\log\log(T/\epsilon)}\bigg)=O\bigg(\norm{\beta}_1t\frac{\log(\norm{\beta}_1t/\epsilon)}{\log\log(\norm{\beta}_1t/\epsilon)}\bigg)
\end{equation}
queries to the operation $\textsc{Select}(\textsc{Exp-}G)=\sum_{l=1}^{L}|l\rangle\langle l|\otimes e^{-i\pi G_l}$.

As mentioned in \cite{FractionalQuery14}, the fractional-query approach can also be used to simulate time-dependent Hamiltonians by replacing \eq{fracquery_timeindep_pf1} with a product-formula decomposition of the time-ordered evolution. However, \cite{FractionalQuery14} only gives a brief discussion of this issue and the claimed complexity has only $L^\infty$ scaling. We now give an improved analysis of this algorithm for the SM model, showing that its query complexity achieves $L^1$-norm scaling.

\begin{theorem}[Fractional-query algorithm with $L^1$-norm scaling (SM)]
	\label{thm:fracquery_sm}
	A $d$-sparse time-dependent Hamiltonian $H(\tau)$ acting on $n$ qubits can be simulated for time $\tau\in[0,t]$ with accuracy $\epsilon$ using
	\begin{equation}
	O\bigg(d^2\norm{H}_{\max,1}\frac{\log(d^2\norm{H}_{\max,\infty}t/\epsilon)}{\log\log(d^2\norm{H}_{\max,\infty}t/\epsilon)}\bigg)
	\end{equation}
	queries to the oracles $\CO_{\text{loc}}$, $\CO_{\text{val}}$.
\end{theorem}
\begin{proof}
	For readability, we assume that $\norm{H}_{\max,1}$, $\norm{H}_{\max,\infty}$, and $\norm{H'}_{\infty,\infty}$ are the norm upper bounds provided to the algorithm.
	We first decompose $\exp_{\mathcal{T}}\bigl({-}i\int_{0}^{t}\mathrm{d}\tau\, H(\tau)\bigr)$ into a product of evolutions of time-independent Hamiltonians $H(kt/r)$, each evolving for time $t/r$. By \lem{generator}, we have
	\begin{equation}
	\begin{aligned}
	\norm{\exp_{\mathcal{T}}\biggl(-i\int_{\frac{kt}{r}}^{\frac{(k+1)t}{r}}\mathrm{d}\tau\,  H(\tau)\biggr)-e^{-i\frac{t}{r}H\big(\frac{kt}{r}\big)}}_\infty
	&\leq\int_{\frac{kt}{r}}^{\frac{(k+1)t}{r}}\mathrm{d}s
	\norm{H(s)-H\bigg(\frac{kt}{r}\bigg)}_\infty\\
	&\leq\int_{\frac{kt}{r}}^{\frac{(k+1)t}{r}}\mathrm{d}s
	\bigg(s-\frac{kt}{r}\bigg)\norm{H'}_{\infty,\infty}\\
	&=\frac{\norm{H'}_{\infty,\infty}t^2}{2r^2},
	\end{aligned}
	\end{equation}
	which implies
	\begin{equation}\label{eq:rmin}
	\norm{\exp_{\mathcal{T}}\biggl(-i\int_{0}^{t}\mathrm{d}s\ H(s)\biggr)-\prod_{k=0}^{r-1}e^{-i\frac{t}{r}H\bigl(\frac{kt}{r}\bigr)}}_\infty
	\leq\frac{\norm{H'}_{\infty,\infty}t^2}{2r}.
	\end{equation}
	To approximate with precision $\epsilon$, it suffices to choose
	\begin{equation}
	\label{eq:r_cond1}
	\begin{aligned}
	r=O\bigg(\frac{\norm{H'}_{\infty,\infty}t^2}{\epsilon}\bigg).
	\end{aligned}
	\end{equation}
Note that we use $O$ here because we can choose $r$ to be the minimum integer satisfying \eq{rmin}, giving an upper bound on the number of steps that suffice to achieve error at most $\epsilon$.
	
	We then decompose the evolution under each time-independent sparse Hamiltonian $H(kt/r)$ for time $t/r$ with precision $O(\epsilon/r)$. By \lem{sparse_decomp}, $H(kt/r)$ can be decomposed into a sum of $\eta=O(d^2\norm{H(kt/r)}_{\max}/\gamma)$ terms $G_j(kt/r)$ such that
	\begin{equation}
	\label{eq:max_norm_bound}
		\norm{H\big(kt/r\big)-\gamma\sum_{j=1}^{\eta}G_j\big(kt/r\big)}_{\max}\leq\sqrt{2}\gamma.
	\end{equation} 
	Furthermore, each $G_j$ is $1$-sparse and Hermitian with eigenvalues $\pm 1$ and the value and location of each non-zero matrix element in $G_j$ can be accessed using $O(1)$ queries to $H$.	
 We choose $\gamma=O\big(\epsilon/td\big)$ so that
	\begingroup
	\allowdisplaybreaks
	\begin{align}
	\norm{e^{-i\frac{t}{r}H\big(\frac{kt}{r}\big)}-e^{-i\frac{t}{r}\gamma\sum_{j=1}^{\eta}G_j\big(\frac{kt}{r}\big)}}_\infty
	&=	\norm{e^{-i\frac{t}{r}\left[H\big(\frac{kt}{r}\big)-\gamma\sum_{j=1}^{\eta}G_j\big(\frac{kt}{r}\big)\right]}-I}_\infty \nonumber \\
	&\le \frac{t}{r}\norm{H(kt/r)-\gamma\sum_{j=1}^{\eta}G_j(kt/r)}_\infty \nonumber \\
	&\le \frac{td}{r}\norm{H(kt/r)-\gamma\sum_{j=1}^{\eta}G_j(kt/r)}_{\max} \nonumber \\
	&\le \frac{ t d }{r}\sqrt{2} \gamma \nonumber 
	= O\left( \frac{td}{r} \frac{\epsilon}{td}\right)\nonumber 
	=O\bigg(\frac{\epsilon}{r}\bigg).
	\end{align}
	\endgroup
	In the third line we have used the inequality between the spectral norm and max norm, in the fourth line we have used the bound on the max norm \eq{max_norm_bound}, and in the fifth line we have used $\gamma=O\big(\epsilon/td\big)$. 
	This implies $\eta=O\big(d^3\norm{H(kt/r)}_{\max}t/\epsilon\big)$ and the fractional query complexity is
	\begin{equation}
	\label{eq:fracquery_seg}
	\eta \frac{t}{r}\gamma=O\bigg(\frac{d^2\norm{H(kt/r)}_{\max}t}{r}\bigg).
	\end{equation}
	We apply the first-order product formula to obtain
	\begin{equation}
	\norm{e^{-i\frac{t}{r}\gamma\sum_{j=1}^{\eta}G_j\big(\frac{kt}{r}\big)}-e^{-i\frac{t}{r}\gamma G_1\big(\frac{kt}{r}\big)}\cdots e^{-i\frac{t}{r}\gamma G_{\eta}\big(\frac{kt}{r}\big)}}_\infty
	=O\bigg(\frac{(d^2\norm{H(kt/r)}_{\max}t)^2}{r^2}\bigg).
	\end{equation}
	Therefore, it is possible to choose $r$ as
	\begin{equation}
	\label{eq:s_def}
	r=O\bigg(\frac{(d^2\norm{H(kt/r)}_{\max}t)^2}{\epsilon}\bigg)
	=O\bigg(\frac{(d^2\norm{H}_{\max,\infty}t)^2}{\epsilon}\bigg),
	\end{equation}
	such that the error of the first-order product-formula decomposition is at most
	\begin{equation}
	\norm{e^{-i\frac{t}{r}\gamma\sum_{j=1}^{\eta}G_j\big(\frac{kt}{r}\big)}-e^{-i\frac{t}{r}\gamma G_1\big(\frac{kt}{r}\big)}\cdots e^{-i\frac{t}{r}\gamma G_{\eta}\big(\frac{kt}{r}\big)}}_\infty
	=O\bigg(\frac{\epsilon}{r}\bigg).
	\end{equation}
By choosing $r$ as the maximum of \eq{r_cond1} and \eq{s_def}, we ensure that the error in each of the $r$ time steps is $O(\epsilon/r)$, so the total error is $O(\epsilon)$.

	Altogether, we find a fractional-query algorithm with total query complexity
	\begin{equation}
	\label{eq:fracquery_total}
	T=O\bigg(\sum_{k=0}^{r-1}\frac{d^2\norm{H(kt/r)}_{\max}t}{r}\bigg)
	\end{equation}
	and error
	\begin{equation}
	\norm{\exp_{\mathcal{T}}\biggl(-i\int_{0}^{t}\mathrm{d}s\ H(s)\biggr)
		-\prod_{k=0}^{r-1}\biggl(e^{-i\frac{t}{rs}\gamma G_1\bigl(\frac{kt}{r}\bigr)}\cdots e^{-i\frac{t}{rs}\gamma G_{\eta}\bigl(\frac{kt}{r}\bigr)}\biggr)^s}_\infty
	\leq\epsilon.
	\end{equation}
	
	We now convert this multi-oracle algorithm to a single-oracle algorithm with the same fractional-query complexity and, with precision $O(\epsilon)$, implement it in the discrete query model.
	The single oracle for the
	standard fractional-query algorithm is now
\begin{equation}\label{eq:oraclet}
\textsc{Select}(\textsc{Exp-}G)=\sum_{k=0}^{r-1} \sum_{l=0}^{\eta(k)-1} |k\rangle\langle k| \otimes |l\rangle\langle l|\otimes e^{-i\pi G_l(kt/r)}.
\end{equation}
This oracle encodes the time-dependence of $H$ in an ancilla. The operators $U_j$ in the fractional-query algorithm then need to increment the time register.

Altogether, we make
	\begin{equation}
	\label{eq:query_seg}
	O\bigg(T\frac{\log(T/\epsilon)}{\log\log(T/\epsilon)}\bigg)
	=O\bigg(\sum_{k=0}^{r-1}\frac{d^2\norm{H(kt/r)}_{\max}t}{r}\frac{\log(d^2\norm{H}_{\max,\infty}t/\epsilon)}{\log\log(d^2\norm{H}_{\max,\infty}t/\epsilon)}\bigg)
	\end{equation}
	discrete queries.
	
	We now show how the query complexity of this approach achieves $L^1$-norm scaling. The intuition is that the total query complexity $\widetilde{O}\big(\sum_{k=0}^{r-1}d^2\norm{H(kt/r)}_{\max}t/r\big)$ should be close to $\widetilde{O}\big(d^2\int_{0}^{t}\mathrm{d}\tau\norm{H(\tau)}_{\max}\big)$ when $r$ is sufficiently large. Specifically,
	\begin{equation}
	\label{eq:query_approx}
	\begin{aligned}
	\bigg|\sum_{k=0}^{r-1}\norm{H\bigg(\frac{kt}{r}\bigg)}_{\max}\frac{t}{r}-\int_{0}^{t}\mathrm{d}\tau\norm{H(\tau)}_{\max}\bigg|
	&\leq\sum_{k=0}^{r-1}\int_{\frac{kt}{r}}^{\frac{(k+1)t}{r}}\mathrm{d}\tau
	\bigg|\norm{H(\tau)}_{\max}-\norm{H\bigg(\frac{kt}{r}\bigg)}_{\max}\bigg|\\
	&\leq\sum_{k=0}^{r-1}\int_{\frac{kt}{r}}^{\frac{(k+1)t}{r}}\mathrm{d}\tau
	\norm{H(\tau)-H\bigg(\frac{kt}{r}\bigg)}_{\max}\\
	&\leq\frac{\norm{H'}_{\max,\infty}t^2}{2r}.
	\end{aligned}
	\end{equation}
	To achieve an additive error of $\delta$, it suffices to choose
	\begin{equation}
	\label{eq:r_cond2}
	r=O\bigg(\frac{\norm{H'}_{\max,\infty}t^2}{\delta}\bigg).
	\end{equation}
	Since $\delta$ can be made arbitrarily close to $0$, we have the total query complexity of
	\begin{equation}
	\begin{aligned}
	&\ O\bigg(d^2\bigg(\int_{0}^{t}\mathrm{d}\tau\norm{H(\tau)}_{\max}+\delta\bigg)\frac{\log(d^2\norm{H}_{\max,\infty}t/\epsilon)}{\log\log(d^2\norm{H}_{\max,\infty}t/\epsilon)}\bigg)\\
	=&\ O\bigg(d^2\norm{H}_{\max,1}\frac{\log(d^2\norm{H}_{\max,\infty}t/\epsilon)}{\log\log(d^2\norm{H}_{\max,\infty}t/\epsilon)}\bigg)
	\end{aligned}
	\end{equation}
as claimed.
\end{proof}

The above analysis shows that the fractional-query algorithm can simulate a time-dependent Hamiltonian with query complexity that scales with the $L^1$-norm.
However, this approach does not directly give a useful result for the gate complexity.
The difficulty arises from the factor of $g$ in \cite[Proof of Theorem 1.1]{FractionalQuery14}, which corresponds to the complexity of applying a sequence of driving operations $U_j$.
These operations need to increment $k$ (indexing the time), as well as $l$, which takes $\eta(k)$ values depending on $k$.
Applying the sequence of operations $U_j$ therefore requires determining new values of $l$ and $k$, which can depend on the sum of $\eta(k)$ over a long sequence of values of $k$.
This will introduce significant gate complexity, so a fast algorithm would require a more efficient procedure for implementing the driving operations.

Instead, we develop other quantum algorithms that achieve $L^1$-norm scaling for not only the query complexity but also the gate complexity. We employ two main techniques: the continuous qDRIFT sampling protocol (\sec{qdrift}) and a rescaling principle for the Schr\"{o}dinger equation (\sec{rescaling}).

\section{Continuous qDRIFT}
\label{sec:qdrift}

We show in \sec{qdrift_universal} that continuous qDRIFT is universal, in the sense that any time-independent Hamiltonian simulable by the algorithm of \cite{Campbell18} can be simulated by our protocol. We then discuss the simulation complexity in both the SM and the LCU models in \sec{qdrift_complexity}.

The continuous qDRIFT protocol also has similarities with the approach of Poulin et al.\  \cite{Pou15} based on Hamiltonian averaging and Monte Carlo sampling, although their approach does not have $L^1$-norm scaling. We give a detailed comparison between these two approaches in \append{poulin}.

\subsection{A classical sampler of time-dependent Hamiltonians}
\label{sec:qdrift_sampler}
Let $H(\tau)$ be a time-dependent Hamiltonian defined for $0\leq\tau\leq t$. For this section only, we relax our requirements on the Hamiltonians: we assume that $H(\tau)$ is nonzero everywhere and is continuous except on a finite number of points. We further suppose that each $H(\tau)$ can be directly exponentiated on a quantum computer. The ideal evolution under $H(\tau)$ for time $t$ is given by $\E(t,0)=\exp_{\mathcal{T}}\big(-i\int_{0}^{t}\mathrm{d}\tau\, H(\tau)\big)$ and the corresponding quantum channel is
\begin{equation}
	\CE(t,0)(\rho)=\E(t,0)\rho\E^\dagger(t,0)
	=\exp_{\mathcal{T}}\biggl(-i\int_{0}^{t}\mathrm{d}\tau\, H(\tau)\biggr)
	\rho
	\exp_{\mathcal{T}}^\dagger\biggl(-i\int_{0}^{t}\mathrm{d}\tau\, H(\tau)\biggr).
\end{equation}
The high-level idea of the sampling algorithm is to approximate the ideal channel by a mixed unitary channel
\begin{equation}
\CU(t,0)(\rho)=\int_{0}^{t}\mathrm{d}\tau\, p(\tau)e^{-i \frac{H(\tau)}{p(\tau)}}\rho e^{i \frac{H(\tau)}{p(\tau)}},
\end{equation}
where $p(\tau)$ is a probability density function defined for $0\leq \tau\leq t$. This channel can be realized by a classical sampling protocol. With a proper choice of $p(\tau)$, this channel approximates the ideal channel and can thus be used for quantum simulation.

We begin with a full definition of $\CU(t,0)$. Inspired by \cite{Campbell18}, we choose $p(\tau)$ to be biased toward those $\tau$ with large $\norm{H(\tau)}_\infty$. A natural choice is
\begin{equation}
	p(\tau):=\frac{\norm{H(\tau)}_{\infty}}{\norm{H}_{\infty,1}}.
\end{equation}
Note that $\CU(t,0)$ is a valid quantum channel (in particular, $p(\tau)$ can never be zero). Furthermore, it can be implemented with unit cost: for any input state $\rho$, we randomly sample a value $\tau$ according to $p(\tau)$ and perform $e^{-i {H(\tau)}/{p(\tau)}}$.
Note also that ${H(\tau)}/{p(\tau)}$ in the exponential implicitly depends on $t$. Indeed, $\norm{H}_{\infty,1}$ includes an integral over time, so $p(\tau)$ decreases with the total evolution time $t$. 
We call this classical sampling protocol and the channel it implements ``continuous qDRIFT''.

This protocol assumes that the spectral norm $\norm{H(\tau)}_\infty$ is known a priori and that we can efficiently sample from the distribution $p(\tau)$. In practice, it is often easier to obtain a spectral-norm upper bound $\Lambda(\tau)\geq\norm{H(\tau)}_\infty$. Such an upper bound can also be used to implement continuous qDRIFT, provided that it has only finitely many discontinuities. Specifically, we define
\begin{equation}
p_\Lambda(\tau):=\frac{\Lambda(\tau)}{\norm{\Lambda}_1},
\end{equation}
so $p_\Lambda(\tau)$ is a probability density function. Using $p_\Lambda$ to implement continuous qDRIFT, we obtain the channel
\begin{equation}
	\CU_\Lambda(t,0)(\rho):=\int_{0}^{t}\mathrm{d}\tau\, p_\Lambda(\tau)e^{-i \frac{H(\tau)}{p_\Lambda(\tau)}}\rho e^{i \frac{H(\tau)}{p_\Lambda(\tau)}},
\end{equation}
whose analysis is similar to that presented here. For readability, we assume that we can efficiently sample from $p(\tau)=\norm{H(\tau)}_{\infty}/\norm{H}_{\infty,1}$ and we analyze $\CU(t,0)$.

We show that continuous qDRIFT approximates the ideal channel with error that depends on the $L^1$-norm.
\begin{theorem}[$L^1$-norm error bound for continuous qDRIFT, short-time version]
	\label{thm:qdrift_short_time}
	Let $H(\tau)$ be a time-dependent Hamiltonian defined for $0\leq\tau\leq t$; assume it is continuous except on a finite number of points and nonzero everywhere. Define $\E(t,0)=\exp_{\mathcal{T}}\bigl(-i\int_{0}^{t}\mathrm{d}\tau H(\tau)\bigr)$ and let $\CE(t,0)(\cdot)=\E(t,0)(\cdot)\E^\dagger(t,0)$ be the corresponding quantum channel. Let $\CU(t,0)$ be the continuous qDRIFT channel
	\begin{equation}
	\CU(t,0)(\rho)=
	\int_0^t\mathrm{d}\tau\, p(\tau)e^{-i \frac{H(\tau)}{p(\tau)}}\rho e^{i \frac{H(\tau)}{p(\tau)}},
	\end{equation}
	where $p(\tau)=\norm{H(\tau)}_{\infty}/\norm{H}_{\infty,1}$. Then
	\begin{equation}
	\norm{\CE(t,0)-\CU(t,0)}_\diamond\leq 4\norm{H}_{\infty,1}^2.
	\label{eq:qdrift_short_time_bound}
	\end{equation}
\end{theorem}

Note that this bound is only useful when $t$ is small enough that the right-hand side of \eq{qdrift_short_time_bound} is less than $1$ (the norm $\norm{H}_{\infty,1}$ involves an integral over $t$, so it increases with $t$).

To prove this theorem, we need a formula that computes the rate at which the evolution operator changes when the Hamiltonian is scaled. To illustrate the idea, consider the degenerate case where the Hamiltonian $H$ is time independent. Then the evolution under $H$ for time $t$ is given by $e^{-itH}$. A direct calculation shows that
\begin{equation}
\frac{\mathrm{d}}{\mathrm{d}s}e^{-itsH}=-itHe^{-itsH},
\end{equation}
so the rate is $-itHe^{-itsH}$ in the time-independent case. This calculation becomes significantly more complicated for a time-dependent Hamiltonian. The following lemma gives an explicit formula for
\begin{equation}
\frac{\mathrm{d}}{\mathrm{d}s}\exp_{\mathcal{T}}\biggl(-i\int_{0}^{t}\mathrm{d}\tau\, sH(\tau)\biggr).
\end{equation}
We sketch the proof of this formula for completeness, but refer the reader to \cite[p.\ 35]{bk:dollard_friedman} for mathematical justifications that are beyond the scope of this paper.

\begin{lemma}[Hamiltonian scaling]
	\label{lem:scaling}
	Let $H(\tau)$ be a time-dependent Hamiltonian defined for $0\leq\tau\leq t$ and assume it has finitely many discontinuities. Denote $\E_s(t,v)=\exp_{\mathcal{T}}\bigl(-i\int_{v}^{t}\mathrm{d}\tau\, sH(\tau)\bigr)$. Then,
	\begin{equation}
	\frac{\mathrm{d}}{\mathrm{d}s}\E_s(t,v)
	=\int_{v}^{t}\mathrm{d}\tau\, 
	\E_s(t,\tau)
	\bigl[-iH(\tau)\bigr]
	\E_s(\tau,v).
	\end{equation}
\end{lemma}
\begin{proof}[Proof sketch]
	We first consider the special case where $H(\tau)$ is continuous in $\tau$. We invoke the variation-of-parameters formula \cite[Theorem 4.9]{bk:knapp} to construct the claimed integral representation for $\frac{\mathrm{d}}{\mathrm{d}s}\E_s(t,v)$. To this end, we need to find a differential equation satisfied by $\frac{\mathrm{d}}{\mathrm{d}t}\frac{\mathrm{d}}{\mathrm{d}s}\E_s(t,v)$ and the corresponding initial condition $\frac{\mathrm{d}}{\mathrm{d}s}\E_s(t,v)\big\rvert_{t=v}$. We differentiate the Schr\"{o}dinger equation $\frac{\mathrm{d}}{\mathrm{d}t}\E_s(t,v)=-isH(t)\E_s(t,v)$ with respect to $s$ to get
	\begin{equation}
	\frac{\mathrm{d}}{\mathrm{d}t}\frac{\mathrm{d}}{\mathrm{d}s}\E_s(t,v)
	=-isH(t)\frac{\mathrm{d}}{\mathrm{d}s}\E_s(t,v)
	-iH(t)\E_s(t,v).
	\end{equation}
	Invoking the variation-of-parameters formula, we find an integral representation
	\begin{equation}
	\label{eq:scale_diff}
	\begin{aligned}
	\frac{\mathrm{d}}{\mathrm{d}s}\E_s(t,v)
	&=\E_s(t,v)\cdot
	\bigg[\frac{\mathrm{d}}{\mathrm{d}s}\E_s(t,v)\Big\rvert_{t=v}\bigg]
	+\E_s(t,v)
	\int_{v}^{t}\mathrm{d}\tau\, 
	\E_s^\dagger(\tau,v)
	\big[-iH(\tau)\big]
	\E_s(\tau,v)\\
	&=\E_s(t,v)\cdot
	\bigg[\frac{\mathrm{d}}{\mathrm{d}s}\E_s(t,v)\Big\rvert_{t=v}\bigg]
	+\int_{v}^{t}\mathrm{d}\tau\, 
	\E_s(t,\tau)
	\big[-iH(\tau)\big]
	\E_s(\tau,v).
	\end{aligned}
	\end{equation}
	It thus remains to find the initial condition $\frac{\mathrm{d}}{\mathrm{d}s}\E_s(t,v)\big\rvert_{t=v}$.
	
	We start from the Schr\"{o}dinger equation $\frac{\mathrm{d}}{\mathrm{d}t}\E_s(t,v)=-isH(t)\E_s(t,v)$ and apply the fundamental theorem of calculus with initial condition $\E_s(v,v)=I$, obtaining the integral representation
	\begin{equation}
	\E_s(t,v)
	=I-is\int_{v}^{t}\mathrm{d}\tau\, H(\tau)\E_s(\tau,v).
	\end{equation}
	Differentiating this equation with respect to $s$ gives
	\begin{equation}
	\frac{\mathrm{d}}{\mathrm{d}s}\E_s(t,v)
	=-i\int_{v}^{t}\mathrm{d}\tau\, H(\tau)\E_s(t,v)
	-is\int_{v}^{t}\mathrm{d}\tau\, H(\tau)\frac{\mathrm{d}}{\mathrm{d}s}\E_s(\tau,v),
	\end{equation}
	which implies
	\begin{equation}
	\label{eq:scale_initial}
	\frac{\mathrm{d}}{\mathrm{d}s}\E_s(t,v)\Big\rvert_{t=v}=0.
	\end{equation}
	Combining \eq{scale_diff} and \eq{scale_initial} establishes the claimed integral representation for $\frac{\mathrm{d}}{\mathrm{d}s}\E_s(t,v)$.
	
	Now consider the case where $H(\tau)$ is piecewise continuous with one discontinuity at $t_1\in[v,t]$. We use the multiplicative property to break the evolution at $t_1$, so that each subevolution is generated by a continuous Hamiltonian. We have 
	\begin{equation}
	\begin{aligned}
		\frac{\mathrm{d}}{\mathrm{d}s}\E_s(t,v)
		&=\frac{\mathrm{d}}{\mathrm{d}s}\big[\E_s(t,t_1)\E_s(t_1,v)\big]\\
		&=\frac{\mathrm{d}}{\mathrm{d}s}\E_s(t,t_1)\cdot\E_s(t_1,v)
		+\E_s(t,t_1)\cdot\frac{\mathrm{d}}{\mathrm{d}s}\E_s(t_1,v)\\
		&=\int_{t_1}^{t}\mathrm{d}\tau\, 
		\E_s(t,\tau)
		\big[-iH(\tau)\big]
		\E_s(\tau,t_1)\cdot\E_s(t_1,v)\\
		&+\E_s(t,t_1)\cdot\int_{0}^{t_1}\mathrm{d}\tau\, 
		\E_s(t_1,\tau)
		\big[-iH(\tau)\big]
		\E_s(\tau,v)\\
		&=\int_{t_1}^{t}\mathrm{d}\tau\, 
		\E_s(t,\tau)
		\big[-iH(\tau)\big]
		\E_s(\tau,v)\\
		&+\int_{v}^{t_1}\mathrm{d}\tau\, 
		\E_s(t,\tau)
		\big[-iH(\tau)\big]
		\E_s(\tau,v)\\
		&=\int_{v}^{t}\mathrm{d}\tau\, 
		\E_s(t,\tau)
		\big[-iH(\tau)\big]
		\E_s(\tau,v).
	\end{aligned}
	\end{equation}
	The general case of finitely many discontinuities follows by induction.
\end{proof}

Note that our argument implicitly assumes the existence of the derivatives and that we can interchange the order of $\frac{\mathrm{d}}{\mathrm{d}s}$ and $\frac{\mathrm{d}}{\mathrm{d}t}$. A rigorous justification of these assumptions is beyond the scope of the paper; we refer the reader to \cite[p.\ 35]{bk:dollard_friedman} for details.

\begin{proof}[Proof of {\thm{qdrift_short_time}}]
	Define two parametrized quantum channels
	\begin{equation}
	\CE_s(t,0)(\rho)=\E_s(t,0)\rho \E_s^\dagger(t,0),\qquad
	\CU_s(t,0)(\rho)=\int_0^t\mathrm{d}\tau\, p(\tau)e^{-i s\frac{H(\tau)}{p(\tau)}}\rho e^{i s\frac{H(\tau)}{p(\tau)}}
	\end{equation}
	and observe that
	\begin{equation}
	\CE_0(t,0)(\rho)=\rho,\quad
	\CE_1(t,0)(\rho)=\CE(t,0)(\rho),\quad
	\CU_0(t,0)(\rho)=\rho,\quad
	\CU_1(t,0)(\rho)=\CU(t,0)(\rho).
	\end{equation}
	To bound the diamond-norm error $\norm{\CE_1(t,0)-\CU_1(t,0)}_\diamond$, we should take a state $\sigma$ on the joint system of the original register and an ancilla register with the same dimension and upper bound $\norm{(\CE_1(t,0)\otimes\mathds{1})(\sigma)-(\CU_1(t,0)\otimes\mathds{1})(\sigma)}_1$. For readability, we instead show how to bound the error $\norm{\CE_1(t,0)(\rho)-\CU_1(t,0)(\rho)}_1$, but 
	the derivation works in exactly the same way for
	the distance $\norm{(\CE_1(t,0)\otimes\mathds{1})(\sigma)-(\CU_1(t,0)\otimes\mathds{1})(\sigma)}_1$ and the resulting bound is the same.
	
	Invoking \lem{scaling}, we have
	\begin{equation}
	\frac{\mathrm{d}}{\mathrm{d}s}\E_s(t,0)\Big\rvert_{s=0}
	=\int_{0}^{t}\mathrm{d}\tau\, 
	\E_s(t,\tau)\Big\rvert_{s=0}
	\big[-iH(\tau)\big]
	\E_s(\tau,0)\Big\rvert_{s=0}
	=-i\int_{0}^{t}\mathrm{d}\tau\, H(\tau).
	\end{equation}
	Thus, the first derivatives of $\CE_s(t,0)(\rho)$ and $\CU_s(t,0)(\rho)$ at $s=0$ agree with each other:
	\begin{equation}
	\begin{aligned}
	\frac{\mathrm{d}}{\mathrm{d}s}\CE_s(t,0)(\rho)\Big\rvert_{s=0}
	&=\bigg[-i\int_{0}^{t}\mathrm{d}\tau\, H(\tau),\rho\bigg]\\
	&=\int_{0}^{t}\mathrm{d}\tau\, p(\tau)\bigg[-i\frac{H(\tau)}{p(\tau)},\rho\bigg]
	=\frac{\mathrm{d}}{\mathrm{d}s}\CU_s(t,0)(\rho)\Big\rvert_{s=0}.
	\end{aligned}
	\end{equation}
	Applying the fundamental theorem of calculus twice, we obtain
	\begin{equation}
	\begin{aligned}
	\CE_1(t,0)(\rho)-\CU_1(t,0)(\rho)
	&=\big(\CE_1(t,0)(\rho)-\CE_0(t,0)(\rho)\big)-\big(\CU_1(t,0)(\rho)-\CU_0(t,0)(\rho)\big)\\
	&=\int_{0}^{1}\mathrm{d}s\int_{0}^{s}\mathrm{d}v\ \frac{\mathrm{d}^2}{\mathrm{d}v^2}\big[\CE_v(t,0)(\rho)-\CU_v(t,0)(\rho)\big] \\
	&=\int_{0}^{1}\mathrm{d}s\int_{0}^{s}\mathrm{d}v\ 
	\Bigg\{\frac{\mathrm{d}^2}{\mathrm{d}v^2}\E_v(t,0)\cdot\rho\cdot \E_v^\dagger(t,0)\\
	&\quad+2\frac{\mathrm{d}}{\mathrm{d}v}\E_v(t,0)\cdot\rho\cdot \frac{\mathrm{d}}{\mathrm{d}v}\E_v^\dagger(t,0)
	+\E_v(t,0)\cdot\rho\cdot \frac{\mathrm{d}^2}{\mathrm{d}v^2}\E_v^\dagger(t,0)\\
	&\quad
	-\int_0^t\mathrm{d}\tau\, p(\tau)e^{-iv \frac{H(\tau)}{p(\tau)}}
	\bigg[-i\frac{H(\tau)}{p(\tau)},\bigg[-i\frac{H(\tau)}{p(\tau)},\rho\bigg]\bigg]
	e^{iv \frac{H(\tau)}{p(\tau)}}\Bigg\}.
	\end{aligned}
	\end{equation}
	By properties of the Schatten norms and the definition $p(\tau)=\norm{H(\tau)}_{\infty}/\norm{H}_{\infty,1}$, we find that
	\begin{equation}
	\begin{aligned}
	&\norm{\CE_1(t,0)(\rho)-\CU_1(t,0)(\rho)}_1\\
	\leq&\int_{0}^{1}\mathrm{d}s\int_{0}^{s}\mathrm{d}v\ 
	\bigg\{2\norm{\frac{\mathrm{d}^2}{\mathrm{d}v^2}\E_v(t,0)}_\infty
	+2\norm{\frac{\mathrm{d}}{\mathrm{d}v}\E_v(t,0)}^2_\infty
	+4\norm{H}_{\infty,1}^2\bigg\}.
	\end{aligned}
	\end{equation}
	\lem{scaling} immediately yields an upper bound on $\norm{\frac{\mathrm{d}}{\mathrm{d}v}\E_v(t,0)}_\infty$:
	\begin{equation}
	\norm{\frac{\mathrm{d}}{\mathrm{d}v}\E_v(t,0)}_\infty\leq\int_{0}^{t}\mathrm{d}\tau\norm{H(\tau)}_\infty=\norm{H}_{\infty,1}.
	\end{equation}
	It thus remains to bound $\norm{\frac{\mathrm{d}^2}{\mathrm{d}v^2}\E_v(t,0)}_\infty$.
	
	Using \lem{scaling} twice, we have
	\begin{equation}
	\begin{aligned}
	\frac{\mathrm{d}^2}{\mathrm{d}v^2}\E_v(t,0)
	&=\frac{\mathrm{d}}{\mathrm{d}v}
	\int_{0}^{t}\mathrm{d}\tau\, 
	\E_v(t,\tau)
	\big[-iH(\tau)\big]
	\E_v(\tau,0)\\
	&=\int_{0}^{t}\mathrm{d}\tau
	\int_{\tau}^{t}\mathrm{d}\tau'\ 
	\E_v(t,\tau')
	\big[-iH(\tau')\big]
	\E_v(\tau',\tau)
	\big[-iH(\tau)\big]
	\E_v(\tau,0)\\
	&\quad +\int_{0}^{t}\mathrm{d}\tau\, 
	\E_v(t,\tau)
	\big[-iH(\tau)\big]
	\int_{0}^{\tau}\mathrm{d}\tau'\ 
	\E_v(\tau,\tau')
	\big[-iH(\tau')\big]
	\E_v(\tau',0),
	\end{aligned}
	\end{equation}
	which implies
	\begin{equation}
	\begin{aligned}
	\norm{\frac{\mathrm{d}^2}{\mathrm{d}v^2}\E_v(\tau,0)}_\infty
	&\leq\int_{0}^{t}\mathrm{d}\tau
	\int_{\tau}^{t}\mathrm{d}\tau'
	\norm{H(\tau')}_\infty\norm{H(\tau)}_\infty
	+\int_{0}^{t}\mathrm{d}\tau
	\int_{0}^{\tau}\mathrm{d}\tau'
	\norm{H(\tau)}_\infty\norm{H(\tau')}_\infty\\
	&=\norm{H}_{\infty,1}^2.
	\end{aligned}
	\end{equation}
	We finally obtain the desired bound
	\begin{equation}
	\norm{\CE_1(t,0)(\rho)-\CU_1(t,0)(\rho)}_1
	\leq\int_{0}^{1}\mathrm{d}s\int_{0}^{s}\mathrm{d}v\ 
	\biggl[2\norm{H}_{\infty,1}^2
	+2\norm{H}_{\infty,1}^2
	+4\norm{H}_{\infty,1}^2\biggr]
	=4\norm{H}_{\infty,1}^2
	\end{equation}
as claimed.
\end{proof}

The above error bound works well for a short-time evolution. When $t$ is large, in order to control the error of simulation, we divide the entire evolution into segments $[t_j,t_{j+1}]$ with $0=t_0<t_1<\cdots<t_r=t$ and apply continuous qDRIFT within each. We employ a variable-time scheme to segment the evolution, so that our $L^1$-norm scaling result can be generalized to a long-time evolution. Specifically, we have:

\begin{theorem}[$L^1$-norm error bound for continuous qDRIFT, long-time version]
	\label{thm:qdrift_long_time}
	Let $H(\tau)$ be a time-dependent Hamiltonian defined for $0\leq\tau\leq t$. Assume that it is continuous except at a finite number of points and nonzero everywhere. Define $\E(t,0)=\exp_{\mathcal{T}}\big(-i\int_{0}^{t}\mathrm{d}\tau\, H(\tau)\big)$ and let $\CE(t,0)(\cdot)=\E(t,0)(\cdot)\E^\dagger(t,0)$ be the corresponding quantum channel. Let $\CU(t,0)$ be the continuous qDRIFT channel
	\begin{equation}
	\CU(t,0)(\rho)=
	\int_0^t\mathrm{d}\tau\, p(\tau)e^{-i \frac{H(\tau)}{p(\tau)}}\rho e^{i \frac{H(\tau)}{p(\tau)}},
	\end{equation}
	where $p(\tau)=\norm{H(\tau)}_{\infty}/\norm{H}_{\infty,1}$. Then, for any positive integer $r$, there exists a division $0=t_0<t_1<\cdots<t_r=t$, such that
	\begin{equation}
		\norm{\CE(t,0)-\prod_{j=0}^{r-1}\CU(t_{j+1},t_j)}_\diamond\leq 4\frac{\norm{H}_{\infty,1}^2}{r}.
	\end{equation}
	To ensure that the simulation error is at most $\epsilon$, it thus suffices to choose
	\begin{equation}
	\label{eq:qdrift_gate_complexity}
	r\geq 4\bigg\lceil\frac{\norm{H}_{\infty,1}^2}{\epsilon}\bigg\rceil.
	\end{equation}
\end{theorem}
\begin{proof}
	The times $t_1,\cdots,t_{r-1}$ are selected as follows. We aim to simulate with accuracy
	\begin{equation}
	4\frac{\norm{H}_{\infty,1}^2}{r^2}
	\end{equation}
	for each segment. To achieve this, we define $t_1,\cdots,t_{r-1}$ so that
	\begin{equation}
	\label{eq:var_time}
	\int_{0}^{t_1}\mathrm{d}\tau\norm{H(\tau)}_\infty
	=\int_{t_1}^{t_2}\mathrm{d}\tau\norm{H(\tau)}_\infty
	=\cdots
	=\int_{t_{r-1}}^{t_r}\mathrm{d}\tau\norm{H(\tau)}_\infty
	=\frac{1}{r}\int_{0}^{t}\mathrm{d}\tau\norm{H(\tau)}_\infty.
	\end{equation}
	The existence of such times is guaranteed by the intermediate value theorem. By telescoping, we find from \thm{qdrift_short_time} that
	\begin{equation}
	\begin{aligned}
	\norm{\CE(t,0)-\prod_{j=0}^{r-1}\CU(t_{j+1},t_j)}_\diamond
	&\leq\sum_{j=0}^{r-1}\norm{\CU(t_{j+1},t_j)-\CE(t_{j+1},t_j)}_\diamond\\
	&\leq\sum_{j=0}^{r-1}4\bigg(\int_{t_{j}}^{t_{j+1}}\mathrm{d}\tau\norm{H(\tau)}_\infty\bigg)^2\\
	&=4r\bigg(\frac{1}{r}\int_{0}^{t}\mathrm{d}\tau\norm{H(\tau)}_\infty\bigg)^2
	=4\frac{\norm{H}_{\infty,1}^2}{r},
	\end{aligned}
	\end{equation}
	which establishes the claimed error bound.
\end{proof}

\subsection{Universality}
\label{sec:qdrift_universal}

We now show that the continuous qDRIFT method introduced above can be applied in the far more general LC model where the Hamiltonian is a sum of time-dependent terms. In this sense it can be regarded as a universal method.

Recall from \sec{prelim_model} that in the general LC model, the Hamiltonian can be expressed as
\begin{equation}
H(\tau)=\sum_{l=1}^{L}H_l(\tau),
\end{equation}
where each $H_l(\tau)$ is continuous and can be efficiently exponentiated on a quantum computer. This includes many familiar models as special cases:
\begin{enumerate}[(i)]
	\item Campbell considered simulating a time-independent Hamiltonian of the form $H=\sum_{l=1}^{L}\alpha_lH_l$, $\norm{H_l}_\infty\leq 1$ \cite{Campbell18}, which is subsumed by the LC model with the time dependence dropped;
	\item if $H(\tau)$ is a time-dependent $d$-sparse Hamiltonian, then \lem{sparse_decomp} shows that it can be decomposed in the form $H(\tau)=\sum_{j=1}^{d^2}H_j(\tau)$, which again belongs to the LC model as the exponentiation of $H_j(\tau)$ can be performed efficiently; and
	\item the LC model is naturally more general than LCU as each summand is not necessarily unitary.
\end{enumerate}

It is not hard to design a classical sampler for time-dependent Hamiltonians in the LC model. A natural choice is
\begin{equation}
\CU(t,0)(\rho):=\sum_{l=1}^{L}\int_0^t\mathrm{d}\tau\, p_l(\tau)e^{-i \frac{H_l(\tau)}{p_l(\tau)}}\rho e^{i \frac{H_l(\tau)}{p_l(\tau)}},
\end{equation}
where $p_l(\tau)$ is the probability distribution
\begin{equation}
p_l(\tau):=\frac{\norm{H_l(\tau)}_\infty}{\norm{H}_{\infty,1,1}}.
\end{equation}
To analyze the performance of this sampler, we adapt the analysis in \thm{qdrift_short_time} and \thm{qdrift_long_time}, which becomes more complicated as we are now sampling a discrete-continuous probability distribution $p_l(\tau)$. Fortunately, a significant amount of effort can be saved with the help of the following universal property.

\begin{theorem}[Universality of continuous qDRIFT]
	\label{thm:universality}
	Let $H(\tau)=\sum_{l=1}^{L}H_l(\tau)$ be a time-dependent Hamiltonian defined for $0\leq\tau\leq t$ that is nonzero everywhere. Assume that each $H_l(\tau)$ is continuous and nonzero everywhere. Define the probability distribution
	\begin{equation}
	p_l(\tau):=\frac{\norm{H_l(\tau)}_\infty}{\norm{H}_{\infty,1,1}}.
	\end{equation}
	Then there exists a time-dependent Hamiltonian $G(\tau)$ defined for $0\leq\tau\leq t$ with finitely many discontinuities, such that the following correspondence holds:
	\begin{enumerate}
		\item $\norm{G}_{\infty,1}=\norm{H}_{\infty,1,1}$.
		\item $\int_{0}^{t}\mathrm{d}\tau\, G(\tau)=\sum_{l=1}^{L}\int_{0}^{t}\mathrm{d}\tau\, H_l(\tau)$.
		\item $\int_0^t\mathrm{d}\tau\, q(\tau)e^{-i \frac{G(\tau)}{q(\tau)}}\rho e^{i \frac{G(\tau)}{q(\tau)}}
		=\sum_{l=1}^{L}\int_0^t\mathrm{d}\tau\, p_l(\tau)e^{-i \frac{H_l(\tau)}{p_l(\tau)}}\rho e^{i \frac{H_l(\tau)}{p_l(\tau)}}$, where we have the probability distribution $q(\tau):=\norm{G(\tau)}_\infty/\norm{G}_{\infty,1}$.
	\end{enumerate}
\end{theorem}

Before presenting the proof, we explain how \thm{universality} can be applied to simulation in the LC model. We expect that the mixed-unitary channel $\sum_{l=1}^{L}\int_0^t\mathrm{d}\tau\, p_l(\tau)e^{-i \frac{H_l(\tau)}{p_l(\tau)}}\rho e^{i \frac{H_l(\tau)}{p_l(\tau)}}$ approximates the ideal evolution with $L^1$-norm scaling as in \thm{qdrift_short_time} and \thm{qdrift_long_time}, but direct analysis would be considerably more complicated. However, universality (Statement 3 of \thm{universality}) shows that this channel is the same as $\int_0^t\mathrm{d}\tau\, q(\tau)e^{-i \frac{G(\tau)}{q(\tau)}}\rho e^{i \frac{G(\tau)}{q(\tau)}}$. Thus, the analysis of \sec{qdrift_sampler} can be applied with the help of \thm{universality}.

\begin{proof}[Proof of {\thm{universality}}]
	We define $G(\tau)$ to be the piecewise Hamiltonian
	\begin{equation}
	\label{eq:piecewise}
	G(\tau)=
	\begin{cases}
	\frac{H_1\big(\frac{\tau}{p_1}\big)}{p_1},\quad &0\leq\tau< p_1 t,\\
	\frac{H_2\big(\frac{\tau-p_1t}{p_2}\big)}{p_2},&p_1t\leq\tau< (p_1+p_2)t,\\
	\qquad\vdots\\
	\frac{H_L\big(\frac{\tau-(p_1+p_2+\cdots+p_{L-1})t}{p_L}\big)}{p_L},&(p_1+p_2+\cdots+p_{L-1})t\leq\tau\leq t,
	\end{cases}
	\end{equation}
	where we use the abbreviation
	\begin{equation}
	p_l:=\norm{p_l}_1=\int_{0}^{t}\mathrm{d}\tau\, p_l(\tau)
	\end{equation}
	for the marginal probability distribution. Statements $1$ and $2$ can both be proved by directly evaluating the integrals
	\begin{equation}
	\begin{aligned}
	\norm{G}_{\infty,1}
	=&\int_{0}^{p_1t}\mathrm{d}\tau\frac{\norm{H_1\big(\frac{\tau}{p_1}\big)}_\infty}{p_1}
	+\int_{p_1t}^{(p_1+p_2)t}\mathrm{d}\tau\frac{\norm{H_2\big(\frac{\tau-p_1t}{p_2}\big)}_\infty}{p_2}\\
	&+\cdots+\int_{(p_1+p_2+\cdots+p_{L-1})t}^{t}\mathrm{d}\tau\frac{\norm{H_L\big(\frac{\tau-(p_1+p_2+\cdots+p_{L-1})t}{p_L}\big)}_\infty}{p_L}\\
	=&\int_{0}^{t}\mathrm{d}\tau\norm{H_1(\tau)}_\infty
	+\int_{0}^{t}\mathrm{d}\tau\norm{H_2(\tau)}_\infty
	+\cdots
	+\int_{0}^{t}\mathrm{d}\tau\norm{H_L(\tau)}_\infty
	=\norm{H}_{\infty,1,1}
	\end{aligned}
	\end{equation}
	and
	\begin{equation}
	\begin{aligned}
	\int_{0}^{t}\mathrm{d}\tau\, G(\tau)
	=&\int_{0}^{p_1t}\mathrm{d}\tau\frac{H_1\big(\frac{\tau}{p_1}\big)}{p_1}+\int_{p_1t}^{(p_1+p_2)t}\mathrm{d}\tau\frac{H_2\big(\frac{\tau-p_1t}{p_2}\big)}{p_2}\\
	&+\cdots+\int_{(p_1+p_2+\cdots+p_{L-1})t}^{t}\mathrm{d}\tau\frac{H_L\big(\frac{\tau-(p_1+p_2+\cdots+p_{L-1})t}{p_L}\big)}{p_L}
	=\sum_{l=1}^{L}\int_{0}^{t}\mathrm{d}\tau\, H_l(\tau).
	\end{aligned}
	\end{equation}
	
	We use Statement $1$ to deduce that
	\begin{equation}
	q(\tau)=\frac{\norm{G(\tau)}_\infty}{\norm{G(\tau)}_{\infty,1}}=
	\begin{cases}
	\frac{\norm{H_1\big(\frac{\tau}{p_1}\big)}_\infty}{p_1\norm{H}_{\infty,1,1}},\quad &0\leq\tau< p_1 t,\\
	\frac{\norm{H_2\big(\frac{\tau-p_1t}{p_2}\big)}_\infty}{p_2\norm{H}_{\infty,1,1}},&p_1t\leq\tau< (p_1+p_2)t,\\
	\qquad\vdots\\
	\frac{\norm{H_L\big(\frac{\tau-(p_1+p_2+\cdots+p_{L-1})t}{p_L}\big)}_\infty}{p_L\norm{H}_{\infty,1,1}},&(p_1+p_2+\cdots+p_{L-1})t\leq\tau\leq t.
	\end{cases}
	\end{equation}
	Therefore,
	\begin{align}
	&\int_{0}^{t}\mathrm{d}\tau\, q(\tau)e^{-i \frac{G(\tau)}{q(\tau)}}\rho e^{i \frac{G(\tau)}{q(\tau)}}\nonumber\\
	&=\int_{0}^{p_1t}\mathrm{d}\tau \frac{\norm{H_1\big(\frac{\tau}{p_1}\big)}_\infty}{p_1\norm{H}_{\infty,1,1}}
	\exp\Bigg(-i \frac{H_1\big(\frac{\tau}{p_1}\big)}{\norm{H_1\big(\frac{\tau}{p_1}\big)}_\infty}\norm{H}_{\infty,1,1}\Bigg)
	\rho
	\exp\Bigg(i \frac{H_1\big(\frac{\tau}{p_1}\big)}{\norm{H_1\big(\frac{\tau}{p_1}\big)}_\infty}\norm{H}_{\infty,1,1}\Bigg)\nonumber\\
	&+\int_{p_1t}^{(p_1+p_2)t}\mathrm{d}\tau \frac{\norm{H_2\big(\frac{\tau-p_1t}{p_2}\big)}_\infty}{p_2\norm{H}_{\infty,1,1}}
	\exp\Bigg(-i\frac{H_2\big(\frac{\tau-p_1t}{p_2}\big)}{\norm{H_2\big(\frac{\tau-p_1t}{p_2}\big)}_\infty}\norm{H}_{\infty,1,1}\Bigg)
	\rho
	\exp\Bigg(i\frac{H_2\big(\frac{\tau-p_1t}{p_2}\big)}{\norm{H_2\big(\frac{\tau-p_1t}{p_2}\big)}_\infty}\norm{H}_{\infty,1,1}\Bigg)\nonumber\\
	&+\cdots+\int_{(p_1+p_2+\cdots+p_{L-1})t}^{t}\mathrm{d}\tau \frac{\norm{H_L\big(\frac{\tau-(p_1+p_2+\cdots+p_{L-1})t}{p_L}\big)}_\infty}{p_L\norm{H}_{\infty,1,1}}\nonumber\\
	&\cdot
	\exp\Bigg(-i\frac{H_L\big(\frac{\tau-(p_1+p_2+\cdots+p_{L-1})t}{p_L}\big)}{\norm{H_L\big(\frac{\tau-(p_1+p_2+\cdots+p_{L-1})t}{p_L}\big)}_\infty}\norm{H}_{\infty,1,1}\Bigg)
	\rho
	\exp\Bigg(i\frac{H_L\big(\frac{\tau-(p_1+p_2+\cdots+p_{L-1})t}{p_L}\big)}{\norm{H_L\big(\frac{\tau-(p_1+p_2+\cdots+p_{L-1})t}{p_L}\big)}_\infty}\norm{H}_{\infty,1,1}\Bigg)\nonumber\\
	&=\sum_{l=1}^{L}\int_0^t\mathrm{d}\tau\, p_l(\tau)e^{-i \frac{H_l(\tau)}{p_l(\tau)}}\rho e^{i \frac{H_l(\tau)}{p_l(\tau)}},
	\end{align}
	which completes the proof of Statement $3$.
\end{proof}

\begin{thmbis}{thm:qdrift_short_time}[$L^1$-norm error bound for continuous qDRIFT (LC), short-time version]
	\label{thm:qdrift_lc_short_time}
	Let $H(\tau)=\sum_{l=1}^{L}H_l(\tau)$ be a time-dependent Hamiltonian defined for $0\leq\tau\leq t$ that is nonzero everywhere. Assume that each $H_l(\tau)$ is continuous and nonzero everywhere. Define $\E(t,0)=\exp_{\mathcal{T}}\big(-i\int_{0}^{t}\mathrm{d}\tau\, H(\tau)\big)$ and let $\CE(t,0)(\cdot)=\E(t,0)(\cdot)\E^\dagger(t,0)$ be the corresponding quantum channel. Let $\CU(t,0)$ be the continuous qDRIFT channel
	\begin{equation}
	\CU(t,0)(\rho):=\sum_{l=1}^{L}\int_0^t\mathrm{d}\tau\, p_l(\tau)e^{-i \frac{H_l(\tau)}{p_l(\tau)}}\rho e^{i \frac{H_l(\tau)}{p_l(\tau)}},
	\end{equation}
	where $p_l(\tau)$ is the probability distribution $p_l(\tau):=\norm{H_l(\tau)}_\infty/\norm{H}_{\infty,1,1}$. Then,
	\begin{equation}
	\norm{\CE(t,0)-\CU(t,0)}_\diamond\leq 4\norm{H}_{\infty,1,1}^2.
	\end{equation}
\end{thmbis}

In the special case where $H=\sum_{l=1}^{L}H_l$ is time independent, our bound reduces to
\begin{equation}
	\norm{\CE(t,0)-\CU(t,0)}_\diamond\leq 4\bigg(\sum_l \norm{H_l}_\infty \bigg)^2t^2.
\end{equation}
This tightens a bound due to Campbell \cite[Eq.~(B12)]{Campbell18} by a multiplicative factor from a tail bound. Note that \cite{Campbell18} considered the distance $\norm{\CE(t,0)-\CU(t,0)}_\diamond/2$, which is different from our definition of the diamond-norm distance $\norm{\CE(t,0)-\CU(t,0)}_\diamond$.

\begin{proof}[Proof of {\thm{qdrift_lc_short_time}}]
	Consider the channel
	\begin{equation}
		\CG(t,0)(\rho):=\int_0^t\mathrm{d}\tau\, q(\tau)e^{-i \frac{G(\tau)}{q(\tau)}}\rho e^{i \frac{G(\tau)}{q(\tau)}},
	\end{equation}
	where $q(\tau):=\norm{G(\tau)}_\infty/\norm{G}_{\infty,1}$ and $G(\tau)$ is defined by \eq{piecewise}. 
    By \thm{universality}, it suffices to bound $\norm{\CE(t,0)-\CG(t,0)}_\diamond$.
	
	Define two parametrized quantum channels
	\begin{equation}
	\CE_s(t,0)(\rho)=\E_s(t,0)\rho \E_s^\dagger(t,0),\qquad
	\CG_s(t,0)(\rho)=\int_0^t\mathrm{d}\tau\, q(\tau)e^{-i s\frac{G(\tau)}{q(\tau)}}\rho e^{i s\frac{G(\tau)}{q(\tau)}}
	\end{equation}
	and observe that
	\begin{equation}
	\CE_0(t,0)(\rho)=\rho\qquad
	\CE_1(t,0)(\rho)=\CE(t,0)(\rho)\qquad
	\CG_0(t,0)(\rho)=\rho\qquad
	\CG_1(t,0)(\rho)=\CG(t,0)(\rho).
	\end{equation}
	For readability, we only consider the trace norm $\norm{\CE_1(t,0)(\rho)-\CG_1(t,0)(\rho)}_1$, whose analysis can be easily adapted to bound $\norm{(\CE_1(t,0)\otimes\mathds{1})(\sigma)-(\CG_1(t,0)\otimes\mathds{1})(\sigma)}_1$ and thus the diamond-norm distance $\norm{\CE_1(t,0)-\CG_1(t,0)}_\diamond$.
	
	By \lem{scaling} and \thm{universality}, we find that the first derivatives of $\CE_s(t,0)(\rho)$ and $\CG_s(t,0)(\rho)$ at $s=0$ agree with each other:
	\begin{equation}
	\begin{aligned}
	\frac{\mathrm{d}}{\mathrm{d}s}\CE_s(t,0)(\rho)\Big\rvert_{s=0}
	=\bigg[-i\int_{0}^{t}\mathrm{d}\tau\, H(\tau),\rho\bigg]
	=\bigg[-i\int_0^t\mathrm{d}\tau\, G(\tau),\rho\bigg]
	=\frac{\mathrm{d}}{\mathrm{d}s}\CG_s(t,0)(\rho)\Big\rvert_{s=0}.
	\end{aligned}
	\end{equation}
	Thus, we can apply the fundamental theorem of calculus twice and obtain
	\begin{equation}
	\begin{aligned}
	&\ \CE_1(t,0)(\rho)-\CG_1(t,0)(\rho)\\
	=&\ \big(\CE_1(t,0)(\rho)-\CE_0(t,0)(\rho)\big)-\big(\CG_1(t,0)(\rho)-\CG_0(t,0)(\rho)\big)\\
	=&\ \int_{0}^{1}\mathrm{d}s\int_{0}^{s}\mathrm{d}v\ \frac{\mathrm{d}^2}{\mathrm{d}v^2}\big[\CE_v(t,0)(\rho)-\CG_v(t,0)(\rho)\big] \\
	=&\ \int_{0}^{1}\mathrm{d}s\int_{0}^{s}\mathrm{d}v\ 
	\Bigg\{\frac{\mathrm{d}^2}{\mathrm{d}v^2}\E_v(t,0)\cdot\rho\cdot \E_v^\dagger(t,0)\\
	&\qquad\qquad\qquad+2\frac{\mathrm{d}}{\mathrm{d}v}\E_v(t,0)\cdot\rho\cdot \frac{\mathrm{d}}{\mathrm{d}v}\E_v^\dagger(t,0)\\
	&\qquad\qquad\qquad+\E_v(t,0)\cdot\rho\cdot \frac{\mathrm{d}^2}{\mathrm{d}v^2}\E_v^\dagger(t,0)\\
	&\qquad\qquad\qquad
	-\int_0^t\mathrm{d}\tau\, q(\tau)e^{-iv \frac{G(\tau)}{q(\tau)}}
	\bigg[-i\frac{G(\tau)}{q(\tau)},\bigg[-i\frac{G(\tau)}{q(\tau)},\rho\bigg]\bigg]
	e^{iv \frac{G(\tau)}{q(\tau)}}\Bigg\},
	\end{aligned}
	\end{equation}
	which implies
	\begin{equation}
	\begin{aligned}
	&\norm{\CE_1(t,0)(\rho)-\CG_1(t,0)(\rho)}_1\\
	\leq&\int_{0}^{1}\mathrm{d}s\int_{0}^{s}\mathrm{d}v\ 
	\bigg\{2\norm{H}_{\infty,1,1}^2
	+2\norm{H}_{\infty,1,1}^2
	+4\norm{G}_{\infty,1}^2\bigg\}
	=4\norm{H}_{\infty,1,1}^2.
	\end{aligned}
	\end{equation}
\end{proof}

\begin{thmbis}{thm:qdrift_long_time}[$L^1$-norm error bound for continuous qDRIFT (LC), long-time version]
	\label{thm:qdrift_lc_long_time}
	Let $H(\tau)=\sum_{l=1}^{L}H_l(\tau)$ be a time-dependent Hamiltonian defined for $0\leq\tau\leq t$ that is nonzero everywhere. Assume that each $H_l(\tau)$ is continuous and nonzero everywhere. Define $\E(t,0)=\exp_{\mathcal{T}}\big(-i\int_{0}^{t}\mathrm{d}\tau\, H(\tau)\big)$ and let $\CE(t,0)(\cdot)=\E(t,0)(\cdot)\E^\dagger(t,0)$ be the corresponding quantum channel. Let $\CU(t,0)$ be the continuous qDRIFT channel
	\begin{equation}
	\CU(t,0)(\rho):=\sum_{l=1}^{L}\int_0^t\mathrm{d}\tau\, p_l(\tau)e^{-i \frac{H_l(\tau)}{p_l(\tau)}}\rho e^{i \frac{H_l(\tau)}{p_l(\tau)}},
	\end{equation}
	where $p_l(\tau)$ is the probability distribution $p_l(\tau):=\norm{H_l(\tau)}_\infty/\norm{H}_{\infty,1,1}$. Then, for any positive integer $r$, there exists a division $0=t_0<t_1<\cdots<t_r=t$, such that
	\begin{equation}
	\norm{\CE(t,0)-\prod_{j=0}^{r-1}\CU(t_{j+1},t_j)}_\diamond\leq 4\frac{\norm{H}_{\infty,1,1}^2}{r}.
	\end{equation}
	To ensure that the simulation error is at most $\epsilon$, it thus suffices to choose
	\begin{equation}
	r\geq 4\bigg\lceil\frac{\norm{H}_{\infty,1,1}^2}{\epsilon}\bigg\rceil.
	\end{equation}
\end{thmbis}

The proof of \thm{qdrift_lc_long_time} follows from~\thm{qdrift_lc_short_time} using the same reasoning as that used to prove \thm{qdrift_long_time}.

\subsection{Complexity of the continuous qDRIFT algorithm}
\label{sec:qdrift_complexity}

As an immediate consequence of universality, we obtain the complexity of the continuous qDRIFT algorithm for simulating time-dependent Hamiltonians in both the SM and the LCU models.
\begin{corollary}[Continuous qDRIFT algorithm with $L^1$-norm scaling (SM)]
	\label{cor:qdrift_sm}
	A $d$-sparse time-dependent Hamiltonian $H(\tau)$ acting on $n$ qubits can be simulated for time $\tau\in[0,t]$ with accuracy $\epsilon$ using
	\begin{equation}
	O\bigg(\frac{d^4\norm{H}_{\max,1}^2}{\epsilon}\bigg)
	\end{equation}
	queries to $\CO_{\text{loc}}$, $\CO_{\text{val}}$ and an additional
	\begin{equation}
	\widetilde{O}\bigg(\frac{d^4\norm{H}_{\max,1}^2}{\epsilon}n\bigg)
	\end{equation}
	gates, assuming that the probability distribution $p_j(\tau):=\norm{H(\tau)}_{\max}/d^2\norm{H}_{\max,1},\ j\in\{1,\ldots,d^2\}$ can be efficiently sampled.
\end{corollary}
\begin{proof}
	For any $\tau\in[0,t]$, \lem{sparse_decomp} shows that $H(\tau)$ admits a decomposition $H(\tau)=\sum_{j=1}^{d^2}H_j(\tau)$, where each $H_j(\tau)$ is $1$-sparse and a query to any $H_j(\tau)$ can be simulated with $O(1)$ queries to $H(\tau)$. We use the continuous qDRIFT algorithm to simulate $H(\tau)=\sum_{j=1}^{d^2}H_j(\tau)$. We estimate
	\begin{equation}
	\begin{aligned}
		\norm{H}_{\infty,1,1}
		&=\sum_{j=1}^{d^2}\int_{0}^{t}\mathrm{d}\tau\, \norm{H_j(\tau)}_\infty
		=\sum_{j=1}^{d^2}\int_{0}^{t}\mathrm{d}\tau\, \norm{H_j(\tau)}_{\max}\\
		&\leq d^2\int_{0}^{t}\mathrm{d}\tau\, \norm{H(\tau)}_{\max}
		=d^2\norm{H}_{\max,1},
	\end{aligned}
	\end{equation}
	where the second equality follows because $H_j(\tau)$ is $1$-sparse, and the inequality follows from \lem{sparse_decomp}. Assuming $\norm{H_j(\tau)}_\infty/\norm{H}_{\infty,1,1}$ can be sampled efficiently, \thm{qdrift_lc_long_time} implies that the algorithm has sample complexity and thus query complexity
	\begin{equation}
		O\bigg(\frac{d^4\norm{H}_{\max,1}^2}{\epsilon}\bigg).
	\end{equation}
	
	For each elementary exponentiation, we initialize a quantum register in the computational basis state $|\tau,j\rangle$ and use it to control the $1$-sparse term we need to simulate. This can be done with gate complexity $\widetilde{O}\big(n\big)$. Since the number of $1$-sparse simulations is the query complexity, we obtain the gate complexity
	\begin{equation}
	\widetilde{O}\bigg(\frac{d^4\norm{H}_{\max,1}^2}{\epsilon}n\bigg)
	\end{equation}
	as claimed.
	
	Our above argument assumes that $\norm{H_j(\tau)}_\infty$ is known a priori and that the distribution $\norm{H_j(\tau)}_\infty/\norm{H}_{\infty,1,1}$ can be efficiently sampled. However, the argument still works if we replace each $\norm{H_j(\tau)}_\infty$ by the upper bound
	\begin{equation}
		\norm{H_j(\tau)}_\infty=\norm{H_j(\tau)}_{\max}\leq\norm{H(\tau)}_{\max},
	\end{equation}
	which means we sample the distribution $p_j(\tau):=\norm{H(\tau)}_{\max}/d^2\norm{H}_{\max,1},\ j\in\{1,\ldots,d^2\}$. The claimed query and gate complexities follow from a similar analysis.
\end{proof}

\begin{corbis}{cor:qdrift_sm}[Continuous qDRIFT algorithm with $L^1$-norm scaling (LCU)]
	\label{cor:qdrift_lcu}
	A time-dependent Hamiltonian with the LCU decomposition $H(\tau)=\sum_{l=1}^{L}\alpha_l(\tau)H_l$, where the controlled exponentiation of each $H_l$ can be performed with $g_e$ gates, can be simulated for time $\tau\in[0,t]$ with accuracy $\epsilon$ with gate complexity
	\begin{equation}
	4\bigg\lceil\frac{\norm{\alpha}_{1,1}^2}{\epsilon}\bigg\rceil g_e,
	\end{equation}
	assuming that the probability distribution $p_l(\tau):=\alpha_l(\tau)/\norm{\alpha}_{1,1}$ can be efficiently sampled.
\end{corbis}
\begin{proof}
	For any $H(\tau)=\sum_{l=1}^{L}\alpha_l(\tau)H_l$, we estimate
	\begin{equation}
		\norm{H}_{\infty,1,1}
		=\sum_{l=1}^{L}\int_{0}^{t}\mathrm{d}\tau\, \alpha_l(\tau)\norm{H_l}_\infty
		=\norm{\alpha}_{1,1}.
	\end{equation}
	The claimed complexity then follows from \thm{qdrift_lc_long_time}.
\end{proof}

\section{Rescaled Dyson-series algorithm}
\label{sec:rescaling}

In this section, we propose a general principle for rescaling the Schr\"{o}dinger equation (\sec{rescaling_principle}). We then apply this principle to improve the Dyson-series algorithm (\sec{rescaling_dscomplexity}) to achieve $L^1$-norm scaling.

\subsection{A rescaling principle for the Schr\"{o}dinger equation}
\label{sec:rescaling_principle}
Let $H(\tau)$ be a time-dependent Hamiltonian defined for $0\leq\tau\leq t$. The evolution under $H(\tau)$ for time $t$ is given by the unitary operator $\E(t,0)=\exp_{\mathcal{T}}\bigl(-i\int_{0}^{t}\mathrm{d}\tau\, H(\tau)\bigr)$, which satisfies the Schr\"{o}dinger equation
\begin{equation}
\frac{\mathrm{d}}{\mathrm{d}t}\E(t,0)=-iH(t)\E(t,0).
\end{equation}

We now propose a rescaling principle that helps to achieve $L^1$-norm scaling.
The goal is to effectively have a Hamiltonian with constant spectral norm.
Recall that for a time-independent Hamiltonian one can multiply the time by a constant and divide the Hamiltonian by the same constant and obtain the same time evolution.
We can achieve something similar with a time-dependent Hamiltonian by rescaling the total evolution time to
\begin{align}
	s=f(t)&:=\int_{0}^{t}\mathrm{d}\tau \norm{H(\tau)}_\infty
\end{align}
and using the rescaled Hamiltonian
\begin{align}
    \widetilde{H}\big(\varsigma\big)&:=H(f^{-1}(\varsigma))/\norm{H(f^{-1}(\varsigma))}_\infty .
\end{align}
From this definition, it is obvious that the Hamiltonian has constant norm, because
\begin{equation}
	\bigl\|{\widetilde{H}\big(\varsigma\big)}\bigr\|_\infty=\norm{\frac{H(f^{-1}(\varsigma))}{\norm{H(f^{-1}(\varsigma))}_\infty}}_\infty=1.
\end{equation}
Moreover, we find that the time-evolution operator satisfies
\begin{equation}
\begin{aligned}
	\frac{\mathrm{d}}{\mathrm{d}s}\E(t,0)
	&=\frac{\mathrm{d}}{\mathrm{d}t}\E(t,0)\cdot\frac{\mathrm{d}t}{\mathrm{d}s}\\
	&=-iH(t)\E(t,0)\cdot\frac{1}{\norm{H(t)}_\infty}\\
	&=-i\widetilde{H}\big(s\big)\E(t,0).
\end{aligned}
\end{equation}
Solving this equation shows that we can obtain exactly the same time-evolution operator using the rescaled time and Hamiltonian:
\begin{equation}
\E(t,0)
=\exp_{\mathcal{T}}\biggl(-i\int_{0}^{s}\mathrm{d}\varsigma\, \widetilde{H}(\varsigma)\biggr).
\end{equation}

We also have the norm equality
\begin{equation}
	s\max_{\varsigma\in[0,s]}\bigl\|{\widetilde{H}(\varsigma)}\bigr\|_\infty
	=s
	=\norm{H}_{\infty,1},
\end{equation}
so any algorithm that simulates the rescaled Hamiltonian $\widetilde{H}(\varsigma)$ with complexity that scales with the $L^\infty$ norm can simulate the original Hamiltonian with $L^1$-norm scaling.

While our above discussion considers the spectral norm $\norm{\cdot}_\infty$, other norms may be used depending on the input model of the Hamiltonian. Indeed, in the analysis for the SM model below we use the max-norm instead of the spectral norm.

Note that it may be hard in practice to compute the exact value of $\norm{H(\tau)}$. However, we can instead use the change of variable
\begin{equation}
	s=f(t):=\int_{0}^{t}\mathrm{d}\tau\, \Lambda(\tau),
\end{equation}
where $\Lambda(\tau)\geq\norm{H(\tau)}$ is any upper bound on the norm that can be efficiently computed.

\subsection{Complexity of the rescaled Dyson-series algorithm}
\label{sec:rescaling_dscomplexity}
In this section, we show how the Dyson-series algorithm \cite{BCCKS14,LW18,Kieferova18} can be rescaled to have $L^1$-norm scaling.
We address this first for the SM model of Hamiltonian access before handling the LCU model (see \sec{prelim_model} for definitions of these models).

Unlike continuous qDRIFT, the rescaled Dyson-series algorithm requires additional oracle access to the input Hamiltonian. Specifically, we need oracles that implement the inverse change-of-variable
\begin{equation}
\label{eq:oracle_var}
\CO_{\text{var}}|\varsigma,z\rangle=|\varsigma,z\oplus f^{-1}(\varsigma)\rangle
\end{equation}
and compute the max-norm
\begin{equation}
\label{eq:oracle_norm}
\CO_{\text{norm}}|\tau,z\rangle=|\tau,z\oplus \norm{H(\tau)}_{\max}\rangle.
\end{equation}
A quantum computer with access to these oracles can simulate time-dependent Hamiltonians with $L^1$-norm scaling.
Note that because $f(\tau)$ increases monotonically, we can use binary search to compute $f^{-1}(\varsigma)$ up to precision $\delta$ using $O(\log(t/\delta))$ queries to $f$, so we expect it to be straightforward to implement the oracle $\mathcal{O}_{\rm var}$ in practice.

\begin{theorem}[Rescaled Dyson-series algorithm with $L^1$-norm scaling (SM)]
	\label{thm:rescaling_sm}
	For $\tau \in [0,t]$, let $H(\tau)$ be a $d$-sparse Hamiltonian acting on $n$ qubits.
	Let $f(t):=\int_{0}^{t}\mathrm{d}\tau \norm{H(\tau)}_{\max}$, and suppose we have an upper bound on the max-norm, denoted $\norm{H\big(f^{-1}(\varsigma)\big)}_{\max}$, that is positive and continuously differentiable.
	Then $H$ can be simulated for time $t$ with accuracy $\epsilon$ using
	\begin{equation}
	\label{eq:dyson_sm}
	O\bigg(d\norm{H}_{\max,1}\frac{\log(d\norm{H}_{\max,1}/\epsilon)}{\log\log(d\norm{H}_{\max,1}/\epsilon)}\bigg)
	\end{equation}
	queries to the oracles $\CO_{\text{loc}}$, $\CO_{\text{val}}$, $\CO_{\text{var}}$, $\CO_{\text{norm}}$ and an additional
	\begin{equation}
	\begin{aligned}\label{eq:gatecount}
	\widetilde{O}\big(d\norm{H}_{\max,1}n\big)
	\end{aligned}
	\end{equation}
	gates. 
\end{theorem}

\begin{proof}
	We simulate the rescaled Hamiltonian $\widetilde{H}(\varsigma):=H\big(f^{-1}(\varsigma)\big)/\norm{H\big(f^{-1}(\varsigma)\big)}_{\max}$ for a total time of $\norm{H}_{\max,1}:=\int_0^t\mathrm{d}\tau\norm{H(\tau)}_{\max}$ using the rescaling function
	\begin{equation}
	f(t):=\int_{0}^{t}\mathrm{d}\tau \norm{H(\tau)}_{\max}.
	\end{equation}
	Following \cite[Theorem 9]{LW18}, we construct a unitary operation that block-encodes
	\begin{equation}
		\sum_{\varsigma\in [0,t/M,2t/M,\ldots, (M-1)t/M]}|\varsigma\rangle\langle\varsigma|\otimes\frac{\widetilde{H}(\varsigma)}{d\bigl\|{\widetilde{H}}\bigr\|_{\max,\infty}}
		=\sum_{\varsigma}|\varsigma\rangle\langle\varsigma|\otimes\frac{\widetilde{H}(\varsigma)}{d}.
	\end{equation}
	This construction is similar to \cite[Lemma 8]{LW18}, except that the Hamiltonian is rescaled. Specifically, we use oracles $\CO_{\text{var}}$ and $\CO_{\text{norm}}$ to implement the transformation
	\begin{equation}
		|\varsigma,0,0\rangle
		\mapsto|\varsigma,f^{-1}(\varsigma),\norm{H(f^{-1}(\varsigma))}_{\max}\rangle,
	\end{equation}
	from which we obtain the rescaled Hamiltonian by querying $\CO_{\text{val}}$ and re-normalizing the result with $\norm{H(f^{-1}(\varsigma))}_{\max}$ to compute $\widetilde{H}_{jk}(\varsigma)$:
	\begin{equation}
	\begin{aligned}
	|f^{-1}(\varsigma),\norm{H(f^{-1}(\varsigma))}_{\max},j,k,z\rangle
	\mapsto|f^{-1}(\varsigma),\norm{H(f^{-1}(\varsigma))}_{\max},j,k,z\oplus\widetilde{H}_{jk}(\varsigma)\rangle.
	\end{aligned}
	\end{equation}
	We then uncompute the ancilla registers storing $f^{-1}(\varsigma)$ and $\norm{H(f^{-1}(\varsigma))}_{\max}$. Overall, this implements a rescaled oracle
	\begin{equation}
	\widetilde{\CO}_{\text{val}}|\varsigma,j,k,z\rangle=|\varsigma,j,k,z\oplus \widetilde{H}_{jk}(\varsigma)\rangle.
	\end{equation}
	The remaining algorithm proceeds as in \cite{LW18}. As the implementation of each $\widetilde{\CO}_{\text{val}}$ requires $O(1)$ queries to the oracles $\CO_{\text{loc}}$, $\CO_{\text{val}}$, $\CO_{\text{var}}$, $\CO_{\text{norm}}$, the overall query complexity is obtained by applying \cite[Theorem 9]{LW18} to the rescaled Hamiltonian, giving query complexity
	\begin{equation}
	    O\bigg(T\frac{\log(T/\epsilon)}{\log\log(T/\epsilon)}\bigg)
	\end{equation}
where
	\begin{equation}
	    T = d \big\|\widetilde{H}\big\|_{\max,\infty} s = d \norm{H}_{\max,1} .
	\end{equation}
Using this expression for $T$ gives the query complexity in \eq{dyson_sm}.

We now analyze the gate complexity.
If the entries of the Hamiltonian are given to within precision
\begin{equation}
    O \left( \frac \epsilon{td} \right),
\end{equation}
then the overall error due to the finite precision is $O(\epsilon)$.
Since the maximum value of any matrix entry of $H$ is $\|H\|_{\max,\infty}$, the number of bits required is
\begin{equation}
   n_p \in  \Theta \left( \log \left(\frac{d\|H\|_{\max,\infty} t} \epsilon\right) \right).
\end{equation}
The implementation involves performing arithmetic on these values, which can be performed with complexity\footnote{In \cite{Kieferova18} and \cite{LW18}, the most complicated operations used are additions, which can be performed with complexity $O(n_p)$. Here we are normalizing the Hamiltonian, so we must also perform multiplication and/or division, for which the straightforward approach has complexity $O(n_p^2)$.
While it is possible to perform multiplication and division with lower asymptotic complexity, such algorithms are only advantageous for very large instances, and do not affect the result as presented in \thm{rescaling_sm}, where logarithmic contributions to the gate complexity are suppressed.} $O(n_p^2)$.
Since this is a logarithmic gate cost for each oracle query,
it gives a contribution to the gate complexity of
$\widetilde{O}\bigl(d\norm{H}_{\max,1}\bigr)$.

The number of time steps is \cite[Corollary 4]{LW18}
\begin{equation}\label{eq:mval}
    M \in \Theta \left( \frac{t}{\alpha\epsilon}\left(\frac{\|{\widetilde H'}\|_{\infty,1}}{t} +\|\widetilde H\|_{\infty,\infty}^2\right) \right)
\end{equation}
where $\alpha = d \|\widetilde H\|_{\max,\infty}$.
The complexity to prepare the time registers is $\log M$ times the query complexity.
We may ignore the second term in \eq{mval}, because it is negligible compared to the complexity of the arithmetic.

We have
\begin{equation}
    \|{\widetilde H'}\|_{\infty,1} = \int_0^s \mathrm{d}\varsigma \, \biggl\| \frac {\mathrm{d}\widetilde H}{\mathrm{d}\varsigma} \biggr\|_\infty .
\end{equation}
Evaluating this derivative, we get
\begin{equation}
		 \frac {\mathrm{d}\widetilde H}{\mathrm{d}\varsigma}
		=\frac{\mathrm{d}\tau}{\mathrm{d}\varsigma}\frac{\mathrm{d}}{\mathrm{d}\tau}\bigg(\frac{H(\tau)}{\norm{H(\tau)}_{\max}}\bigg)
		=\frac{\mathrm{d}\tau}{\mathrm{d}\varsigma} \left( \frac{H'\big(\tau\big)}
		{\norm{H\big(\tau\big)}_{\max}}-\frac{H(\tau)\norm{H\big(\tau\big)}_{\max}'}{\norm{H(\tau)}_{\max}^2}\right) ,
\end{equation}
so we obtain
\begin{align}
\|{\widetilde H'}\|_{\infty,1} &= \int_0^t \mathrm{d}\tau \left\|  \frac{H'\big(\tau\big)}
		{\norm{H(\tau)}_{\max}}-\frac{H\big(\tau\big)\norm{H\big(\tau\big)}_{\max}'}{\norm{H(\tau)}_{\max}^2} \right\|_\infty \nonumber \\
	&\le \frac{\norm{H'(\tau)}_{\infty,1}}
	{\min_{\tau\in [0,t]}\norm{H(\tau)}_{\max}} + \frac{\norm{\norm{H(\tau)}_{\infty}\norm{H(\tau)}'_{\max}}_1}
	{\min_{\tau\in [0,t]}\norm{H(\tau)}_{\max}^2} \nonumber \\
	&\le \frac{\norm{H'(\tau)}_{\infty,1}}
	{\min_{\tau\in [0,t]}\norm{H(\tau)}_{\max}} + \frac{\norm{H(\tau)}_{\infty,2}\norm{\norm{H(\tau)}'_{\max}}_2}
	{\min_{\tau\in [0,t]}\norm{H(\tau)}_{\max}^2} .
\end{align}
The gate complexity of the preparation of the time registers is $\log M$ times the query complexity, where we have shown that
\begin{equation}
    M \in \Theta \left( \frac{\|{\widetilde H'}\|_{\infty,1}}{\epsilon d\|H\|_{\max, \infty}} \right),
\end{equation}
where $\|{\widetilde H'}\|_{\infty,1}$ is polynomial in norms of $H$ and its derivative.
Since this a logarithmic cost, the contribution to the complexity from preparation of the time registers is $\widetilde{O}\bigl(d\norm{H}_{\max,1}\bigr)$.

The remaining contribution to the gate complexity comes from acting on the system itself.
The cost of this is $O(n)$ for each of the oracle queries, which gives gate complexity $\widetilde{O}\big(d\norm{H}_{\max,1}n\big)$ (this is the dominant cost in \eq{gatecount}).
\end{proof}

Thus, the rescaled Dyson-series algorithm can simulate time-dependent Hamiltonians in the SM model with $L^1$-norm scaling.  Next we turn our attention to the LCU model. For an input Hamiltonian $H(\tau)=\sum_{l=1}^{L}\alpha_l(\tau)H_l$, this approach assumes quantum access to the coefficient oracle
\begin{equation}
	O_{\text{coeff}}|\tau,l,z\rangle=|\tau,l,z\oplus\alpha_l(\tau)\rangle,
\end{equation}
in contrast to the continuous qDRIFT which only needs classical access. Given a classical circuit that computes the coefficients $\alpha_l(\tau)$, we can express it as a sequence of elementary gates and construct a corresponding quantum circuit with the same gate complexity. In our analysis, we ignore the implementation details and count the number of uses of the quantum oracle $\CO_{\text{coeff}}$. The definitions of $\CO_{\text{var}}$ and $\CO_{\text{norm}}$ are similar to the SM case, except that the norm $\norm{H(\tau)}_{\max}$ is replaced by $\norm{\alpha(\tau)}_{\infty}$.

\begin{thmbis}{thm:rescaling_sm}[Rescaled Dyson-series algorithm with $L^1$-norm scaling (LCU)]
	\label{thm:rescaling_lcu}
	For $\tau \in [0,t]$, let $H$ be a time-dependent Hamiltonian with the decomposition $H(\tau)=\sum_{l=1}^{L}\alpha_l(\tau)H_l$, where each controlled $H_l$ can be performed with $g_c$ gates. Let $f(t):=\int_{0}^{t}\mathrm{d}\tau \norm{\alpha(\tau)}_{\infty}$, and suppose we have an upper bound on the $\ell_\infty$ norm of the coefficients, denoted $\norm{\alpha\big(f^{-1}(\varsigma)\big)}_{\infty}$, that is continuously differentiable.
	Then $H$ can be simulated for time $t$ with accuracy $\epsilon$ using
	\begin{equation}
		O\bigg(L\norm{\alpha}_{\infty,1}\frac{\log(L\norm{\alpha}_{\infty,1}/\epsilon)}{\log\log(L\norm{\alpha}_{\infty,1}/\epsilon)}\bigg)
	\end{equation}
	queries to the oracles $\CO_{\text{coeff}}$, $\CO_{\text{var}}$, $\CO_{\text{norm}}$ and an additional
	\begin{equation}
	\begin{aligned}
	\widetilde{O}\bigl(\norm{\alpha}_{\infty,1}L^2g_c\bigr)
	\end{aligned}
	\end{equation}
	gates.
\end{thmbis}
\begin{proof}
	We simulate the rescaled Hamiltonian $\widetilde{H}(\varsigma):=H\big(f^{-1}(\varsigma)\big)/\norm{\alpha\big(f^{-1}(\varsigma)\big)}_{\infty}$ for time $\norm{\alpha}_{\infty,1}$ using the rescaling function
	\begin{equation}
	f(t):=\int_{0}^{t}\mathrm{d}\tau \norm{\alpha(\tau)}_{\infty}.
	\end{equation}
	The rescaled Hamiltonian takes the form
	\begin{equation}
		\widetilde{H}(\varsigma)=\sum_{l=1}^{L}\widetilde{\alpha}_l(f^{-1}(\varsigma))H_l,
	\end{equation}
	where $\widetilde{\alpha}_l(\tau):=\alpha_l(\tau)/\norm{\alpha(\tau)}_{\infty}$. Therefore, we have a global upper bound on the absolute value of the coefficients
	\begin{equation}
		\norm{\widetilde{\alpha}}_{\infty,\infty}\leq 1.
	\end{equation}
	The remaining construction is similar to \cite[Section V C]{Kieferova18}, except that the Hamiltonian is rescaled. Specifically, we use oracles $\CO_{\text{var}}$ and $\CO_{\text{norm}}$ to implement the transformation
	\begin{equation}
	|\varsigma,0,0\rangle
	\mapsto|\varsigma,f^{-1}(\varsigma),\norm{\alpha(f^{-1}(\varsigma))}_{\infty}\rangle,
	\end{equation}
	from which we obtain the rescaled coefficients by querying $\CO_{\text{coeff}}$ and doing arithmetic, giving
	\begin{equation}
	\begin{aligned}
	|f^{-1}(\varsigma),\norm{\alpha(f^{-1}(\varsigma))}_{\infty},l,z\rangle
	\mapsto
	|f^{-1}(\varsigma),\norm{\alpha(f^{-1}(\varsigma))}_{\infty},l,z\oplus\widetilde{\alpha}_l(\varsigma)\rangle.
	\end{aligned}
	\end{equation}
	We then uncompute the ancilla registers storing $f^{-1}(\varsigma)$ and $\norm{\alpha(f^{-1}(\varsigma))}_{\infty}$. Overall, this implements a rescaled oracle
	\begin{equation}
	\widetilde{\CO}_{\text{coeff}}|\varsigma,l,z\rangle=|\varsigma,l,z\oplus\widetilde{\alpha}_l(\varsigma)\rangle.
	\end{equation}
	The remaining algorithm proceeds as in \cite{Kieferova18}. As the implementation of each $\widetilde{\CO}_{\text{coeff}}$ requires $O(1)$ queries to the oracles $\CO_{\text{coeff}}$, $\CO_{\text{var}}$, $\CO_{\text{norm}}$, the overall query complexity is obtained by applying \cite[Theorem 2]{Kieferova18} to the rescaled Hamiltonian. The analysis of the gate complexity proceeds along similar lines to that of \thm{rescaling_sm}. The multiplicative factor of $Lg_c$ is the cost of implementing the $\textsc{Select}$ operation
	\begin{equation}
		\textsc{Select}(H)=\sum_{l=1}^{L}|l\rangle\langle l|\otimes H_l .
	\end{equation}
That complexity may be obtained from \cite[Lemma G.7]{CMNRS18}, or the unary iteration procedure of \cite[Section III A]{BGBWMPFN18}.
\end{proof}

\section{Applications to chemistry and scattering theory}
\label{sec:scatter}

There are numerous cases in physics where one needs to simulate time-dependent quantum systems.  Indeed, the pulse sequences that constitute individual quantum gates or adiabatic sweeps are described by time-dependent Hamiltonians.
Here, we look at the particular case of simulating semi-classical scattering of molecules within a chemical reaction as an example of time-dependent Hamiltonian dynamics~\cite{tully1998mixed,fernandez2006modeling}.

Chemical scattering problems involve colliding reagents.
As the molecules move closer, the electronic configuration changes due to strengthening Coulomb interactions, which is ultimately responsible for either the reagents forming a bond or flying apart depending on the initial conditions and the nature of the reagents.
In the non-relativistic case, the Hamiltonian for two colliding atoms $A$ and $B$ at positions $x_A$ and $x_B$, respectively, and $M$ electrons with positions $x_m$ for $m=1,\ldots,M$, can be expressed as
\begin{align}
H &= H_{\text{nuc}} + H_{\text{elec}}\nonumber\\
H_{\text{nuc}}&=\frac{p_A^2}{2m_A} +\frac{p^2_B}{2m_B} + \frac{Z_A Z_B}{|x_A - x_B|} \nonumber\\
H_{\text{elec}}&= \sum_{m=1}^M \frac{p_m^2}{2m_e}- \sum_{m=1}^M \frac{Z_A}{|x_m-x_A|} - \sum_{m=1}^M \frac{Z_B}{|x_m-x_B|} + \sum_{m < m'} \frac{1}{|x_m -x_{m'}|}.
\end{align}
Here $p_m= [[p_m]_x, [p_m]_y, [p_m]_z]$ and $x_m=[[x_m]_x, [x_m]_y, [x_m]_z]$ are three-dimensional vectors of operators, whereas the corresponding nuclear terms (such as $x_A$ and $x_B$) are three-dimensional vectors of scalars.  We further define $|x_m -x_m'|$ to be the operator $$|x_m -x_{m'}| :=\sqrt{([x_m]_x - [x_{m'}]_x)^2 + ([x_m]_y - [x_{m'}]_y)^2+([x_m]_z - [x_{m'}]_z)^2}.$$

The wave function can be thought of as having a nuclear as well as an electronic component. First, we assume that the nuclear and the electronic wave functions are decoupled~\cite{gerber1982time,whittier1999quantum}:
\begin{equation}
\psi(x_A,x_B,x_1,\ldots,x_M;t)\approx \psi_{\text{nuc}}(x_A,x_B;t) \psi_{\text{elec}}(x_1,\ldots,x_M;t).
\end{equation}
This approximation is justified by the fact that the nuclear mass is substantially greater than the electronic mass.
We then follow the time-dependent self-consistent field (TDSCF) approximation, which further treats $x_A$ and $x_B$ as classical degrees of freedom $x_A(t)$ and $x_B(t)$ with conjugate momenta ($p_A(t), p_B(t)$).
This simplification is justified by Ehrenfest's theorem, which states that for a sufficiently narrow quantum wave packet, the equation of motion for the centroid follows the classical trajectory (to leading order in $\hbar$).
Under this approximation, the electronic dynamics satisfy
\begin{equation}
i\partial_t \ket{\psi_{\text{elec}}(t)} = \left( \sum_{m=1}^M \frac{p_m^2}{2m_e}-\frac{Z_A}{|x_m-x_A(t)|} -  \frac{Z_B}{|x_m-x_B(t)|} + \sum_{m < m'} \frac{1}{|x_m -x_{m'}|} \right)\ket{\psi_{\text{elec}}(t)},
\end{equation}
where we have suppressed the implicit dependence of the electronic wave function on $x_1,\ldots,x_M$.
The equation of motion for the two nuclear positions in the time-dependent self-consistent field approximation within the Ehrenfest method is given by the Hamilton-Jacobi equation:
\begin{align}
\partial_t[p_A(t)]_i &= - \partial_{[x_A]_i} \bra{\psi_{\text{elec}}(t)}H_{\text{elec}} \ket{\psi_{\text{elec}}(t)} - \partial_{[x_A]_i} H_{\text{nuc}}(t),\nonumber\\
\partial_t[x_A(t)]_i &=  \partial_{[p_A]_i} H_{\text{nuc}}(t),\label{eq:HJ}
\end{align}
and similarly for $x_B.$  The function $H_{\text{nuc}}(t)$ here is simply the Hamiltonian $H_{\text{nuc}}$ with the classical substitutions $x_A \rightarrow x_A(t)$, $p_A \rightarrow p_A(t)$ and similarly for $x_B$ and $p_B$.  Similarly, we define $H_{\text{elec}}(t)$ to be the electronic Hamiltonian under this classical substitution.

The evolution in the Ehrenfest method is governed by a pair of tightly coupled quantum and classical dynamical equations, wherein the full Schr\"{o}dinger equation only needs to be solved to understand part of the dynamics for the system.  Indeed, as the Born-Oppenheimer approximation instantaneously holds under the above approximations, we can further express the electronic dynamics within a second-quantized framework with respect to a basis of molecular orbitals as

\begin{equation}
H_{\text{elec}}(t) = \sum_{pq} h_{pq}(t) a^\dagger_p a_q + \frac{1}{2}\sum_{pqrs} h_{pqrs}(t) a^\dagger_pa^\dagger_q a_r a_s,
\end{equation}
where for some basis of orthonormal molecular orbitals ${\psi_p(\vec{x};t)}$ (which are implicitly time dependent if these basis functions are chosen to be functions of the nuclear positions, as would be appropriate for an atomic orbital basis),
\begin{align}
h_{pq}(t) &= \iint \mathrm{d}_{\vec{x}_1}\mathrm{d}_{\vec{x}_2}\psi_p^*(\vec{x}_1;t) \left(\sum_{m=1}^M \frac{p_m^2}{2m_e}-\frac{Z_A}{|x_m-x_A(t)|} -  \frac{Z_B}{|x_m-x_B(t)|}\right)\psi_q(\vec{x}_2;t)\\
h_{pqrs}(t) &= \iiiint \mathrm{d}_{\vec x_1}\mathrm{d}_{\vec x_2}\mathrm{d}_{\vec x_3}\mathrm{d}_{\vec x_4}\psi_p^*(\vec x_1;t) \psi_q^*(\vec x_2;t) \sum_{m < m'} \frac{1}{|x_m -x_{m'}|} \psi_r(\vec x_3;t)\psi_s(\vec x_4;t).\end{align}
Thus under the above approximations, the dynamics that need to be simulated take the form of a standard second quantized simulation of chemistry, except the Hamiltonian is time dependent.  The generalization of this to multiple nuclei is similarly straightforward, with the summation over two nuclear positions replaced by summation over all $L$ positions.

Consider the case where two reagents move towards each other from distant points with large momenta. To get an intuitive understanding of this evolution, it is instructive to examine the case of two molecules colliding using a classical force field.  This will give us an expression that is qualitatively accurate for $x_A(t)$ and $x_B(t)$.  To do this, we use a Lennard-Jones potential to model the interaction between two helium nuclei.  The potential as a function of separation between the nuclei $r(t) = |x_A(t) -x_B(t)|$ is assumed to be of the form $V(r)=\epsilon\left(\frac{r_m^{12}}{r^{12}}-\frac{2 r_m^6}{r^6}  \right)$, where $\epsilon \approx 10$ K and $r_m \approx 2.6$ \AA~\cite{talu2001reference}. Setting the initial radial velocity to be approximately the root mean square (RMS) velocity of helium at $25\,\degree\mathrm{C}$, we solve the classical equations of motion to find the trajectory shown in~\fig{trajectory}.

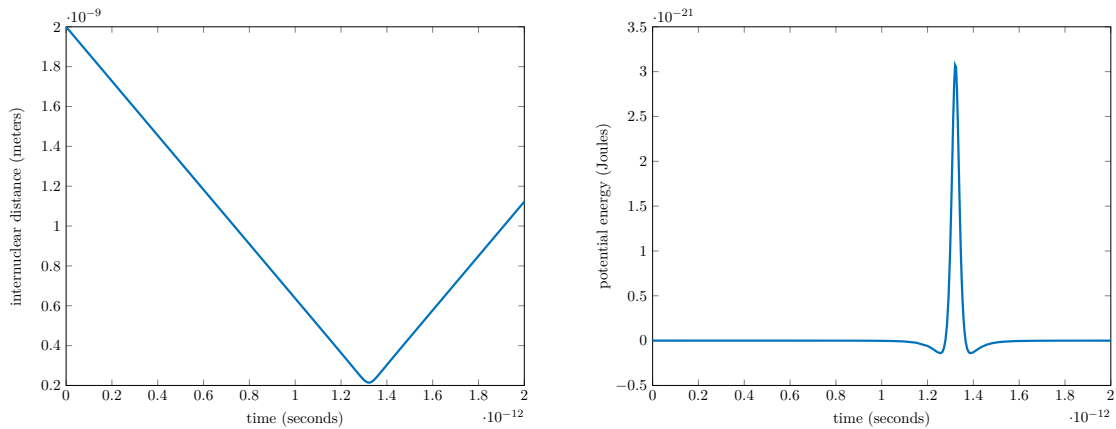
\begin{figure}[t]
	\centering
\begin{tikzpicture}[scale=0.525]
\begin{axis}[%
xlabel = time (seconds), 
ylabel = internuclear distance (meters),
width=4.521in,
height=3.566in,
at={(0.758in,0.481in)},
scale only axis,
xmin=0,
xmax=2e-12,
ymin=2e-10,
ymax=2e-09,
axis background/.style={fill=white},
every axis plot/.append style={ultra thick}
]
\addplot [color=mycolor1, forget plot]
table[row sep=crcr]{%
	0	2e-09\\
	5e-14	1.93184999696794e-09\\
	1e-13	1.86369998685748e-09\\
	1.5e-13	1.79554996760868e-09\\
	2e-13	1.72739993644371e-09\\
	2.5e-13	1.6592498895958e-09\\
	3e-13	1.59109982224485e-09\\
	3.5e-13	1.52294972721486e-09\\
	4e-13	1.45479959429022e-09\\
	4.5e-13	1.38664940863744e-09\\
	5e-13	1.31849915125234e-09\\
	5.5e-13	1.2503487905919e-09\\
	6e-13	1.18219827805748e-09\\
	6.5e-13	1.11404753379456e-09\\
	7e-13	1.04589646209988e-09\\
	7.5e-13	9.77744868343269e-10\\
	8e-13	9.09592411938907e-10\\
	8.5e-13	8.41438336055993e-10\\
	9e-13	7.73282075696726e-10\\
	9.5e-13	7.05121576618258e-10\\
	1e-12	6.36952261991363e-10\\
	1.05e-12	5.68751290693006e-10\\
	1.1e-12	5.00528444566762e-10\\
	1.15e-12	4.32226238349944e-10\\
	1.2e-12	3.63653918855061e-10\\
	1.21453426008449e-12	3.43629864808454e-10\\
	1.22906852016898e-12	3.23528988571234e-10\\
	1.24360278025347e-12	3.03339620072528e-10\\
	1.25813704033796e-12	2.83093498483341e-10\\
	1.26484617862794e-12	2.73754307278116e-10\\
	1.27155531691793e-12	2.64449948009562e-10\\
	1.27826445520791e-12	2.55237194338059e-10\\
	1.2849735934979e-12	2.46218040932867e-10\\
	1.29168273178788e-12	2.37536622629269e-10\\
	1.29839187007786e-12	2.29506240532663e-10\\
	1.30510100836785e-12	2.22594337702493e-10\\
	1.31181014665783e-12	2.17372068169345e-10\\
	1.31602567315468e-12	2.15250144318447e-10\\
	1.32024119965152e-12	2.14210564210077e-10\\
	1.32445672614836e-12	2.14329001418702e-10\\
	1.32867225264521e-12	2.15587856874763e-10\\
	1.33288777914205e-12	2.17894185851448e-10\\
	1.33710330563889e-12	2.21106896269637e-10\\
	1.34131883213574e-12	2.25062359466784e-10\\
	1.34553435863258e-12	2.29591808005001e-10\\
	1.34945915407049e-12	2.34187947362864e-10\\
	1.35338394950839e-12	2.39055880482273e-10\\
	1.3573087449463e-12	2.44120327352532e-10\\
	1.3612335403842e-12	2.49323824909099e-10\\
	1.36706219637564e-12	2.57221355131953e-10\\
	1.37289085236707e-12	2.65243434463405e-10\\
	1.37871950835851e-12	2.73329679318237e-10\\
	1.38454816434995e-12	2.81443985978484e-10\\
	1.39720952068707e-12	2.99088169172525e-10\\
	1.4098708770242e-12	3.16696720579333e-10\\
	1.42253223336132e-12	3.34239702663289e-10\\
	1.43519358969845e-12	3.51719633598105e-10\\
	1.45926250738848e-12	3.84830925153818e-10\\
	1.48333142507851e-12	4.1783049941497e-10\\
	1.50740034276854e-12	4.50762313528507e-10\\
	1.53146926045857e-12	4.83653548339168e-10\\
	1.58146926045857e-12	5.51907527703049e-10\\
	1.63146926045857e-12	6.20110502576033e-10\\
	1.68146926045857e-12	6.88294337751152e-10\\
	1.73146926045857e-12	7.56470673575796e-10\\
	1.78146926045857e-12	8.24640039701906e-10\\
	1.83146926045857e-12	8.92806513300362e-10\\
	1.88146926045857e-12	9.60971409884483e-10\\
	1.93146926045857e-12	1.02913540792992e-09\\
	1.94860194534393e-12	1.05249189055369e-09\\
	1.96573463022928e-12	1.07584831475201e-09\\
	1.98286731511464e-12	1.09920468888961e-09\\
	2e-12	1.12256101998942e-09\\
};
\end{axis}
\end{tikzpicture}
\hspace{0.5cm}
\begin{tikzpicture}[scale=0.525]
\begin{axis}[%
xlabel = time (seconds), 
 ylabel = potential energy (Joules),
width=4.521in,
height=3.566in,
at={(0.758in,0.481in)},
scale only axis,
xmin=0,
xmax=2e-12,
ymin=-5e-22,
ymax=3.5e-21,
axis background/.style={fill=white},
every axis plot/.append style={ultra thick}
]
\addplot [color=mycolor1, forget plot]
  table[row sep=crcr]{%
0	-2.45788410592268e-27\\
5e-14	-3.02622021639525e-27\\
1e-13	-3.75391439396308e-27\\
1.5e-13	-4.69412533273978e-27\\
2e-13	-5.92081257463483e-27\\
2.5e-13	-7.53818305656968e-27\\
3e-13	-9.69507649996447e-27\\
3.5e-13	-1.26072439610286e-26\\
4e-13	-1.65924862259064e-26\\
4.5e-13	-2.21271870696981e-26\\
5e-13	-2.99392619436143e-26\\
5.5e-13	-4.11646572711223e-26\\
6e-13	-5.76178395725825e-26\\
6.5e-13	-8.22730481631516e-26\\
7e-13	-1.20148562829305e-25\\
7.5e-13	-1.7999123918975e-25\\
8e-13	-2.77618509633264e-25\\
8.5e-13	-4.42861747489152e-25\\
9e-13	-7.347656646707e-25\\
9.5e-13	-1.27694695392704e-24\\
1e-12	-2.34576636092833e-24\\
1.05e-12	-4.60910268491424e-24\\
1.1e-12	-9.82653193344972e-24\\
1.15e-12	-2.30937872399911e-23\\
1.2e-12	-5.98159563044226e-23\\
1.21453426008449e-12	-7.93617624687783e-23\\
1.22906852016898e-12	-1.03818306498229e-22\\
1.24360278025347e-12	-1.29646788032431e-22\\
1.25813704033796e-12	-1.40054130136295e-22\\
1.26484617862794e-12	-1.26138612578223e-22\\
1.27155531691793e-12	-8.49688750352065e-23\\
1.27826445520791e-12	4.85412881384555e-24\\
1.2849735934979e-12	1.77573439364405e-22\\
1.29168273178788e-12	4.83635831286041e-22\\
1.29839187007786e-12	9.7741451504936e-22\\
1.30510100836785e-12	1.67115565944801e-21\\
1.31181014665783e-12	2.45026256878973e-21\\
1.31602567315468e-12	2.85224160121871e-21\\
1.32024119965152e-12	3.07084207177981e-21\\
1.32445672614836e-12	3.04516765934848e-21\\
1.32867225264521e-12	2.78441760729718e-21\\
1.33288777914205e-12	2.35972200543493e-21\\
1.33710330563889e-12	1.86643022455648e-21\\
1.34131883213574e-12	1.38654935625778e-21\\
1.34553435863258e-12	9.70652740175688e-22\\
1.34945915407049e-12	6.58262991564544e-22\\
1.35338394950839e-12	4.16382149432098e-22\\
1.3573087449463e-12	2.35926225689341e-22\\
1.3612335403842e-12	1.05327548171589e-22\\
1.36706219637564e-12	-2.02343162163834e-23\\
1.37289085236707e-12	-9.00160585820489e-23\\
1.37871950835851e-12	-1.24980279187874e-22\\
1.38454816434995e-12	-1.38933387299492e-22\\
1.39720952068707e-12	-1.34153298363989e-22\\
1.4098708770242e-12	-1.12823106153523e-22\\
1.42253223336132e-12	-9.02591556006383e-23\\
1.43519358969845e-12	-7.08540841910952e-23\\
1.45926250738848e-12	-4.42732679739709e-23\\
1.48333142507851e-12	-2.8013432101179e-23\\
1.50740034276854e-12	-1.81290946006427e-23\\
1.53146926045857e-12	-1.20218567885183e-23\\
1.58146926045857e-12	-5.51111144277711e-24\\
1.63146926045857e-12	-2.75292122436219e-24\\
1.68146926045857e-12	-1.47556271345057e-24\\
1.73146926045857e-12	-8.38192908057419e-25\\
1.78146926045857e-12	-4.99771526843144e-25\\
1.83146926045857e-12	-3.10427362016718e-25\\
1.88146926045857e-12	-1.99676367161381e-25\\
1.93146926045857e-12	-1.32375069697969e-25\\
1.94860194534393e-12	-1.15702200653595e-25\\
1.96573463022928e-12	-1.01428242529418e-25\\
1.98286731511464e-12	-8.91667985723935e-26\\
2e-12	-7.8600082214069e-26\\
};
\end{axis}
\end{tikzpicture}%
\caption{Trajectory for two helium atoms colliding head on and interacting according to a Lennard-Jones potential with an initial separation of $20$ nm and a velocity of $1350$ m/s.\label{fig:trajectory}}
\end{figure}

From \fig{trajectory}, we find that the interaction appears as a brief but intense kick between the two systems. As a result, the norm of the Hamiltonian changes dramatically throughout the evolution and we expect simulation algorithms with $L^1$-norm scaling to be advantageous over previous approaches. We leave a detailed study of such an advantage as a subject for future work.

\section{Discussion}
\label{sec:discussion}

We have shown that a time-dependent Hamiltonian $H(\tau)$ can be simulated for the time interval $0\leq\tau\leq t$ with gate complexity that scales according to the $L^1$ norm $\int_{0}^{t}\mathrm{d}\tau\norm{H(\tau)}$. We designed new algorithms based on classical sampling and improved the previous Dyson-series approach to achieve this scaling. This is a polynomial speedup in terms of the norm dependence, an advantage that can be favorable in practice. In particular, our result has potential applications to simulating scattering processes in quantum chemistry. Our analysis also matches the intuition that the difficulty of simulating a quantum system should depend on the norm of the Hamiltonian instantaneously. This dual interpretation suggests that the $L^1$-norm dependence of our result cannot be significantly improved. However, further speedup might be possible if we know a priori the energy range of the initial state, as is suggested in \cite{Pou15,LC172}.

The rescaled Dyson-series approach is nearly optimal with respect to all parameters of interest. Indeed, a lower bound of $\Omega\big(d\norm{H}_{\max}t+\frac{\log(1/\epsilon)}{\log\log(1/\epsilon)}\big)$ was given in \cite[Theorem 2]{BCK15} for simulating time-independent sparse Hamiltonians, which of course also holds for the more general time-dependent case. The query complexity \eq{dyson_sm} of the rescaled Dyson-series approach matches this dependence on $d\norm{H}_{\max} t$ and on $\epsilon$, except that it scales as the product of the two terms instead of the sum (so, as in all quantum simulation algorithms prior to the advent of quantum signal processing \cite{LC17}, it does not achieve the optimal tradeoff between $t$ and $\epsilon$). However, this approach requires computing the rescaling function \eq{oracle_var} and the Hamiltonian norm \eq{oracle_norm} in quantum superposition, which may introduce large overhead in practice. In comparison, continuous qDRIFT relies on classical sampling and may be better suited to near-term simulation. Its complexity has no dependence on the parameter $L$ in the LCU decomposition (\cor{qdrift_lcu}), which is advantageous for Hamiltonians consisting of many terms.

For most of our analysis, we have assumed that the Hamiltonian $H(\tau)$ is continuously differentiable. This assumption can be relaxed to allow finitely many discontinuities. Indeed, if $H(\tau)$ is discontinuous at the times $0=t_0<t_1<\cdots<t_r=t$ but otherwise continuously differentiable, we may divide the evolution into $r$ segments and apply a time-dependent Hamiltonian simulation algorithm within each time interval $[t_j,t_{j+1}]$. For the Dyson-series approach, the complexity depends linearly on the $L^1$ norm, so concatenation gives a simulation of the entire evolution with $L^1$-norm scaling. The assumptions about the Hamiltonian can be even further relaxed: the continuous qDRIFT algorithm works properly provided only that $H(\tau)$ is Lebesgue integrable. Further discussion of this point is beyond the scope of this paper, and we refer the reader to \cite{bk:dollard_friedman} for details.

Our analysis can also be adapted to simulate time-dependent Hamiltonians that have countably many zeros. Indeed, since the equation $H(\tau)=0$ has at most countably many solutions, we can find $c\in\R$ such that $H(\tau)+cI$ is nonzero everywhere. Then, $\exp_{\mathcal{T}}\bigl(-i\int_{0}^{t}\mathrm{d}\tau\, (H(\tau)+cI)\bigr)=e^{-ict}\exp_{\mathcal{T}}\bigl(-i\int_{0}^{t}\mathrm{d}\tau\, H(\tau)\bigr)$,
so the result is only off by a global phase. Note that this assumption can be completely dropped if we use continuous qDRIFT: we define the exceptional set
\begin{equation}
\mathcal{B}_0:=p^{-1}(0)=\{\tau:p(\tau)=0\}=\{\tau:\norm{H(\tau)}_\infty=0\}=\{\tau:H(\tau)=0\}
\end{equation}
and redefine $\CU(t,0)$ as
\begin{equation}
\label{eq:qdrift}
\CU(t,0)(\rho):=
\int_{[0,t]\backslash\mathcal{B}_0}\mathrm{d}\tau\ p(\tau)e^{-i \frac{H(\tau)}{p(\tau)}}\rho e^{i \frac{H(\tau)}{p(\tau)}},\qquad
p(\tau):=\frac{\norm{H(\tau)}_\infty}{\norm{H}_{\infty,1}}.
\end{equation}
We note that $\CU(t,0)$ is a valid quantum channel and can be implemented with unit cost. Indeed, for any input state $\rho$, we randomly sample a value $\tau$ according to $p(\tau)$ and perform $e^{-i {H(\tau)}/{p(\tau)}}$ if $\tau\in[0,t]\backslash\mathcal{B}_0$, and the identity operation otherwise. This implements
\begin{equation}
\begin{aligned}
\int_{[0,t]\backslash\mathcal{B}_0}\mathrm{d}\tau\ p(\tau)e^{-i \frac{H(\tau)}{p(\tau)}}\rho e^{i \frac{H(\tau)}{p(\tau)}}
+\int_{\mathcal{B}_{0}}\mathrm{d}\tau\ p(\tau)\rho
=\CU(t,0)(\rho).
\end{aligned}
\end{equation}
The remaining analysis proceeds as in \sec{qdrift}.

The qDRIFT protocol that we analyzed here only achieves first-order accuracy. It is natural to ask if sampling a different probability distribution could lead to an algorithm with better performance. The answer seems to be ``no'' if we only use a univariate distribution. To see this, consider the discrete case where $H=\sum_{l=1}^{L}H_l$ is a Hamiltonian consisting of $L$ terms. We sample according to a probability vector $p\in[0,1]^L$. Upon getting outcome $l$, we perform the unitary $e^{-it H_l/p_l}$. Effectively, we implement the quantum channel $\CU(t)(\rho):=\sum_{l=1}^{L}p_l e^{-it \frac{H_l}{p_l}}\rho e^{it \frac{H_l}{p_l}}$,
which is a first-order approximation to the ideal evolution $\CE(t)(\rho):=e^{-it\sum_{l=1}^{L}H_l}\rho e^{it\sum_{l=1}^{L}H_l}$. In particular, the difference between $\CU(t)(\rho)$ and $\CE(t)(\rho)$ admits an integral representation
\begin{equation}
\begin{aligned}
\CU(t)(\rho)-\CE(t)(\rho)
=&\int_{0}^{t}\mathrm{d}u\int_{0}^{u}\mathrm{d}v\ 
\left\{\sum_{l=1}^{L}p_l e^{-iv \frac{H_l}{p_l}}\bigg[-i\frac{H_l}{p_l},\bigg[-i\frac{H_l}{p_l},\rho\bigg]\bigg] e^{iv \frac{H_l}{p_l}}\right. \\ 
& \left.-e^{-iv\sum_{l=1}^{L}H_l}
\bigg[-i\sum_{l=1}^{L}H_l,\bigg[-i\sum_{l=1}^{L}H_l,\rho\bigg]\bigg]
e^{iv\sum_{l=1}^{L}H_l}\right\}.
\end{aligned}
\end{equation}
To estimate the diamond-norm error $\norm{\CU(t)-\CE(t)}_\diamond$, we take $\sigma$ to be a state on the joint system of the original register and an ancilla register with the same dimension. We compute
\begin{equation}
\begin{aligned}
\norm{(\CU(t)\otimes\mathds{1})(\sigma)-(\CE(t)\otimes\mathds{1})(\sigma)}_1
\leq&\int_{0}^{t}\mathrm{d}u\int_{0}^{u}\mathrm{d}v\ 
\left\{\sum_{l=1}^{L}p_l\norm{\bigg[-i\frac{H_l}{p_l}\otimes\mathds{1},\bigg[-i\frac{H_l}{p_l}\otimes\mathds{1},\sigma\bigg]\bigg]}_1\right. \\
&\qquad\qquad\quad \left.+\norm{\bigg[-i\sum_{l=1}^{L}H_l\otimes\mathds{1},\bigg[-i\sum_{l=1}^{L}H_l\otimes\mathds{1},\sigma\bigg]\bigg]}_1\right\} \\
\leq&\ 2t^2\Bigg(\sum_{l=1}^{L}\frac{\norm{H_l}_\infty^2}{p_l}+\norm{H}_{\infty,1}^2\Bigg).
\end{aligned}
\end{equation}
By Jensen's inequality,
\begin{equation}
\sum_{l=1}^{L}\frac{\norm{H_l}_\infty^2}{p_l}
=\sum_{l=1}^{L}p_l\Bigg(\frac{\norm{H_l}_\infty}{p_l}\Bigg)^2
\geq\Bigg(\sum_{l=1}^{L}p_l\frac{\norm{H_l}_\infty}{p_l}\Bigg)^2
=\norm{H}_{\infty,1}^2,
\end{equation}
with equality if and only if all ${\norm{H_l}_\infty}/{p_l}$ are equal, implying that the probability distribution $p_l:={\norm{H_l}_\infty}/{\norm{H}_{\infty,1}}$ is optimal. A similar optimality result holds for continuous qDRIFT (though the proof is more involved).

However, this does not preclude the existence of a higher-order qDRIFT protocol using more complicated sampling \cite{OWC20}. For example, besides the basic evolutions $e^{-it H_l/p_l}$, one could evolve under commutators $[H_j,H_k]$ or anticommutators $\{H_j,H_k\}$. We could also use a multivariate distribution and correlate different steps of the qDRIFT protocol. For future work, it would be interesting to find a higher-order protocol, or prove that such a protocol cannot exist.

The fractional-query algorithm described in \sec{prelim_fracquery} provides a natural approach to simulating time-dependent Hamiltonians whose query complexity scales with the $L^1$-norm. While we believe such a scaling also holds for the gate complexity, it would be highly nontrivial to give an explicit implementation. In any case, the fractional-query approach is streamlined by the Dyson-series approach and the latter can be rescaled to achieve $L^1$-norm scaling.

The rescaling principle that we proposed can potentially be applied to improve other quantum simulation algorithms. For example, we can use the product-formula algorithm \cite{WBHS10} to simulate the rescaled Hamiltonian  $\widetilde{H}(\varsigma):=H\big(f^{-1}(\varsigma)\big)/\norm{H\big(f^{-1}(\varsigma)\big)}_\infty$ for time $s=\norm{H}_{\infty,1}$. The difficulty here is that the derivative of the rescaled Hamiltonian can be larger than the original one, making the rescaled algorithm perform worse. We leave a thorough study of this issue as a subject for future work.

Finally, it would be interesting to identify further applications of our $L^1$-norm scaling result, such as to designing new quantum algorithms and to improving the performance of quantum chemistry simulation. It might also be of interest to demonstrate these approaches experimentally, for applications such as implementing adiabatic algorithms with quantum circuits.

\section*{Acknowledgements}
We thank Ryan Babbush, Earl Campbell, Andr\'{a}s Gily\'{e}n, M\'{a}ria Kieferov\'{a}, Robin Kothari, Tong\-yang Li, Guang Hao Low, Jarrod R. McClean, Yuval Sanders, and Leonard Wossnig for helpful discussions. DWB is funded by Australian Research Council Discovery Projects DP160102426 and DP190102633. This work was also supported in part by the Army Research Office (MURI award W911NF-16-1-0349), the Canadian Institute for Advanced Research, the National Science Foundation (grant 1813814), and the U.S. Department of Energy, Office of Science, Office of Advanced Scientific Computing Research, Quantum Algorithms Teams and Quantum Testbed Pathfinder programs. In addition, YS is supported by the Google Ph.D.\ Fellowship program. XW acknowledges support from the Department of Defense. NW is supported in part by a Google Quantum Research Award.

\bibliographystyle{plainnat}
\bibliography{L1NormScalingSimulation}

\appendix
\section{Continuous qDRIFT and Hamiltonian averaging}
\label{append:poulin}

Poulin, Qarry, Somma, and Verstraete developed an algorithm for time-dependent Hamiltonian simulation based on techniques of Hamiltonian averaging and Monte Carlo estimation \cite{PQSV11}. In this section, we discuss the relation between their algorithm and our continuous qDRIFT.

Let $H(\tau)$ be a time-dependent Hamiltonian defined for $0\leq\tau\leq t$. Assume that $H(\tau)$ is continuous, nonzero everywhere, and efficiently simulable for each particular $\tau$. Then, Poulin et al.'s approach simulates $H(\tau)$ for time $\tau\in[0,t]$ in two steps: (i) they replace the evolution $\exp_{\mathcal{T}}\big({-}i\int_{0}^{t}\mathrm{d}\tau\, H(\tau)\big)$ by an ordinary matrix exponential $e^{-itH_{\text{av}}}$ of the average Hamiltonian $H_{\text{av}}:=\frac{1}{t}\int_{0}^{t}\mathrm{d}\tau\, H(\tau)$ with an error that scales like $O\big((t\norm{H}_{\infty,\infty})^2\big)$; (ii) they further implement $\int_{0}^{t}\mathrm{d}\tau\, H(\tau)$ with Monte Carlo estimation by picking $m$ random times and approximating $\int_{0}^{t}\mathrm{d}\tau\ H(\tau)\approx \frac{1}{m}\sum_{k=1}^{m}H(\tau_k)$ with error $O\big(t\norm{H}_{\infty,\infty}/\sqrt{m}\big)$, the result of which is further approximated by product formulas.

The approach of \cite{PQSV11} is essentially a sampling-based algorithm and thus similar in spirit to our continuous qDRIFT, except for a notable difference: their algorithm scales with the $L^\infty$ norm instead of the $L^1$ norm. Unfortunately, this drawback cannot be remedied merely by a better analysis of the same algorithm. Indeed, they use a uniform distribution to pick random times during the Monte Carlo estimation. This sampling ignores the instantaneous norm $\norm{H(\tau)}_\infty$ of the Hamiltonian and therefore the resulting algorithm cannot scale with the $L^1$ norm $\int_{0}^{t}\mathrm{d}\tau \norm{H(\tau)}_\infty$.

Instead, continuous qDRIFT uses a probability distribution that biases toward those times with larger instantaneous norm. In \sec{qdrift}, we proved that such a sampling gives a direct simulation of time-dependent Hamiltonians with complexity that scales with the $L^1$ norm. We now give an indirect implementation: (i') we show in \append{poulin_average} that the error of replacing $\exp_{\mathcal{T}}\big(-i\int_{0}^{t}\mathrm{d}\tau\, H(\tau)\big)$ by an ordinary matrix exponential of $H_{\text{av}}$ scales like $O\big(\norm{H}_{\infty,1}^2\big)$, improving the analysis of \cite{PQSV11}; (ii') we further prove in \append{poulin_qdrift} that the average Hamiltonian can be simulated by continuous qDRIFT with $L^1$-norm scaling. Combining these two steps, we see that the Monte Carlo estimation approach of \cite{PQSV11} is superseded by continuous qDRIFT.

\subsection{Hamiltonian averaging}
\label{append:poulin_average}
Let $H(\tau)$ be a time-dependent Hamiltonian defined for $0\leq\tau\leq t$ and assume that it is continuous and nonzero everywhere. Define
\begin{equation}
	\E(s,0):=\exp_{\mathcal{T}}\bigg(-i\int_{0}^{s}\mathrm{d}\tau\, H(\tau)\bigg),\qquad
	\E_{\text{av}}(s):=e^{-isH_{\text{av}}},
\end{equation}
where $H_{\text{av}}:=\frac{1}{t}\int_{0}^{t}\mathrm{d}\tau\, H(\tau)$ is the average Hamiltonian. Our goal is to bound the distance between $\E(s,0)$ and $\E_{\text{av}}(s)$ at $s=t$. Using the initial condition $\E(0,0)=\E_{\text{av}}(0)=I$, we have
\begin{equation}
\begin{aligned}
\norm{\E(t,0)-\E_{\text{av}}(t)}_\infty
=&\norm{\E_{\text{av}}^\dagger(t) \E(t,0)-I}_\infty
=\norm{\int_{0}^{t}\mathrm{d}s\frac{\mathrm{d}}{\mathrm{d}s}\big[\E_{\text{av}}^\dagger(s) \E(s,0)\big]}_\infty.
\end{aligned}
\end{equation}

By the Schr\"{o}dinger equation
\begin{equation}
\frac{\mathrm{d}}{\mathrm{d}s}\E_{\text{av}}(s)=-iH_{\text{av}}\E_{\text{av}}(s),\qquad
\frac{\mathrm{d}}{\mathrm{d}s}\E(s,0)=-iH(s)\E(s,0),
\end{equation}
we obtain
\begin{equation}
\begin{aligned}
\int_{0}^{t}\mathrm{d}s\frac{\mathrm{d}}{\mathrm{d}s}\big[\E_{\text{av}}^\dagger(s) \E(s,0)\big]
&=\int_{0}^{t}\mathrm{d}s \left\{\E_{\text{av}}^\dagger(s) \big[iH_{\text{av}}\big] \E(s,0)
+\E_{\text{av}}^\dagger(s) \big[{-iH(s)}\big] \E(s,0)\right\} \\
&=\frac{1}{t}\int_{0}^{t}\mathrm{d}s\int_{0}^{t}\mathrm{d}\tau \left\{ \E_{\text{av}}^\dagger(s) \big[iH(\tau)\big] \E(s,0)
+\E_{\text{av}}^\dagger(s) \big[{-iH(s)}\big] \E(s,0)\right\} \\
&=\frac{1}{t}\int_{0}^{t}\mathrm{d}s\int_{0}^{t}\mathrm{d}\tau \left\{\E_{\text{av}}^\dagger(s) \big[iH(\tau)\big] \E(s,0)
+\E_{\text{av}}^\dagger(\tau) \big[{-iH(\tau)}\big] \E(\tau,0)\right\},
\end{aligned}
\end{equation}
which implies, by telescoping, that
\begin{equation}
\begin{aligned}
\norm{\int_{0}^{t}\mathrm{d}s\frac{\mathrm{d}}{\mathrm{d}s}\big[\E_{\text{av}}^\dagger(s) \E(s,0)\big]}_\infty
&\leq\frac{1}{t}\int_{0}^{t} \mathrm{d}s\int_{0}^{t}\mathrm{d}\tau\, \big( \norm{\E_{\text{av}}(s)-\E_{\text{av}}(\tau)}_\infty\norm{H(\tau)}_\infty\\
&\qquad\qquad\qquad\quad+\norm{\E(s,0)-\E(\tau,0)}_\infty\norm{H(\tau)}_\infty\big).
\end{aligned}
\end{equation}

By the fundamental theorem of calculus, the first term of the integrand can be bounded as
\begin{equation}
\norm{\E_{\text{av}}(s)-\E_{\text{av}}(\tau)}_\infty
\leq\norm{H_{\text{av}}}_\infty|s-\tau|
\leq\frac{1}{t}\int_{0}^{t}\mathrm{d}u\norm{H(u)}_\infty|s-\tau|
\leq\norm{H}_{\infty,1}.
\end{equation}
To handle the second term, we use \lem{generator}. Observe that the generator of $\E(s,0)$ is $H(u),\ 0\leq u\leq s$, whereas the generator of $\E(\tau,0)$ is $H(u),\ 0\leq u\leq \tau$. So they only differ on the interval $\big[\min\{s,\tau\},\max\{s,\tau\}\big]$. Consequently,
\begin{equation}
\norm{\E(s,0)-\E(\tau,0)}_\infty
\leq\int_{\min\{s,\tau\}}^{\max\{s,\tau\}}\mathrm{d}u\norm{H(u)}_\infty
\leq\int_{0}^{t}\mathrm{d}u\norm{H(u)}_\infty=\norm{H}_{\infty,1}.
\end{equation}
Altogether, we have
\begin{equation}
\begin{aligned}
\norm{\E(t,0)-\E_{\text{av}}(t)}_\infty
&=\norm{\int_{0}^{t}\mathrm{d}s\frac{\mathrm{d}}{\mathrm{d}s}\big[\E_{\text{av}}(s)^\dagger \E(s,0)\big]}_\infty\\
&\leq\frac{1}{t}\int_{0}^{t}\mathrm{d}s\int_{0}^{t}\mathrm{d}\tau\ \left(\norm{\E_{\text{av}}(s)-\E_{\text{av}}(\tau)}_\infty\norm{H(\tau)}_\infty \right. \\
&\qquad\qquad\qquad\quad \left. +\norm{\E(s,0)-\E(\tau,0)}_\infty\norm{H(\tau)}_\infty  \right) \\
&\leq 2\norm{H}_{\infty,1}^2 .
\end{aligned}
\end{equation}

\begin{theorem}[Hamiltonian simulation by averaging (spectral-norm distance)]
	\label{thm:average_spectral}
	Let $H(\tau)$ be a time-dependent Hamiltonian defined for $0\leq\tau\leq t$ and assume that it is continuous and nonzero everywhere. Define $\E(t,0):=\exp_{\mathcal{T}}\big(-i\int_{0}^{t}\mathrm{d}\tau\, H(\tau)\big)$ and $\E_{\text{av}}(t):=e^{-itH_{\text{av}}}$, where $H_{\text{av}}:=\frac{1}{t}\int_{0}^{t}\mathrm{d}\tau\, H(\tau)$ is the average Hamiltonian. Then,
	\begin{equation}
	\norm{\E(t,0)-\E_{\text{av}}(t)}_\infty
	\leq 2\norm{H}_{\infty,1}^2.
	\end{equation}
\end{theorem}

The above bound on the spectral-norm error can be converted to a bound on the diamond-norm error using \lem{diamond_bound}.

\begin{thmbis}{thm:average_spectral}[Hamiltonian simulation by averaging (diamond-norm distance)]
	\label{thm:average_diamond}
	Let $H(\tau)$ be a time-dependent Hamiltonian defined for $0\leq\tau\leq t$ and assume that it is continuous and nonzero everywhere. Define unitary operators $\E(t,0):=\exp_{\mathcal{T}}\big(-i\int_{0}^{t}\mathrm{d}\tau\ H(\tau)\big)$, $\E_{\text{av}}(t):=e^{-itH_{\text{av}}}$ and let $\CE(t,0)(\cdot):=\E(t,0)(\cdot)\E(t,0)^\dagger$, $\CE_{\text{av}}(t)(\cdot)=\E_{\text{av}}(t)(\cdot)\E_{\text{av}}(t)^\dagger$ be the corresponding channels. Then,
	\begin{equation}
	\norm{\CE(t,0)-\CE_{\text{av}}(t)}_\diamond
	\leq 4\norm{H}_{\infty,1}^2.
	\end{equation}
\end{thmbis}

\subsection{Implementing Hamiltonian averaging by continuous qDRIFT}
\label{append:poulin_qdrift}
Let $H(\tau)$ be a time-dependent Hamiltonian defined for $0\leq\tau\leq t$ and assume that it is continuous and nonzero everywhere. We have showed that the ideal evolution can be approximated by an evolution under the average Hamiltonian with error that scales with the $L^1$ norm. We now show that such a Hamiltonian averaging can be implemented by continuous qDRIFT, again with $L^1$-norm scaling. This improves over the algorithm of \cite{PQSV11} which scales with the $L^\infty$ norm.

\begin{theorem}[Hamiltonian averaging by continuous qDRIFT]
	\label{thm:average_qdrift}
	Let $H(\tau)$ be a time-dependent Hamiltonian defined for $0\leq\tau\leq t$ and assume that it is continuous and nonzero everywhere. Define $\E_{\text{av}}(t):=e^{-itH_{\text{av}}}$ and let $\CE_{\text{av}}(t)(\cdot)=\E_{\text{av}}(t)(\cdot)\E_{\text{av}}(t)^\dagger$ be the corresponding channels. Let $\CU(t,0)$ be the continuous qDRIFT channel
	\begin{equation}
	\CU(t,0)(\rho)=
	\int_0^t\mathrm{d}\tau\ p(\tau)e^{-i \frac{H(\tau)}{p(\tau)}}\rho e^{i \frac{H(\tau)}{p(\tau)}},
	\end{equation}
	where $p(\tau)=\norm{H(\tau)}_{\infty}/\norm{H}_{\infty,1}$. Then,
	\begin{equation}
	\norm{\CE_{\text{av}}(t)-\CU(t,0)}_\diamond\leq 4\norm{H}_{\infty,1}^2.
	\end{equation}
\end{theorem}

Note that by applying the triangle inequality to \thm{qdrift_short_time} and \thm{average_diamond}, we obtain
\begin{equation}
\norm{\CE_{\text{av}}(t)-\CU(t,0)}_\diamond\leq 8\norm{H}_{\infty,1}^2.
\end{equation}
\thm{average_qdrift} improves the constant prefactor from $8$ to $4$.

\begin{proof}[Proof of {\thm{average_qdrift}}]
	We parametrize the two channels $\CE_{\text{av}}(t)$, $\CU(t)$ and define
	\begin{equation}
	\begin{aligned}
	\CE_{\text{av},u}(t)(\rho):=e^{-i u\int_{0}^{t}\mathrm{d}\tau H(\tau)}\rho e^{i u\int_{0}^{t}\mathrm{d}\tau H(\tau)},\quad
	\CU_u(t,0)(\rho):=\int_0^t\mathrm{d}\tau\ p(\tau)e^{-i u\frac{H(\tau)}{p(\tau)}}\rho e^{i u\frac{H(\tau)}{p(\tau)}}.
	\end{aligned}
	\end{equation}
	Since $\CE_{\text{av},0}(t)(\rho)=\rho$, $\CE_{\text{av},1}(t)(\rho)=\CE_{\text{av}}(t)(\rho)$, $\CU_0(t,0)(\rho)=\rho$, and $\CU_1(t,0)(\rho)=\CU(\rho)$, the first derivative of $\CE_{\text{av},u}(t)(\rho)$ and $\CU_u(t,0)(\rho)$ agrees with each other at $u=0$
	\begin{equation}
	\begin{aligned}
	\frac{\mathrm{d}}{\mathrm{d}u}\CE_{\text{av},u}(t)(\rho)\bigg|_{u=0}
	=\bigg[{-i}\int_{0}^{t}\mathrm{d}\tau\, H(\tau),\rho\bigg]
	=\frac{\mathrm{d}}{\mathrm{d}u}\CU_u(t,0)(\rho)\bigg|_{u=0}.
	\end{aligned}
	\end{equation}
	Applying the fundamental theorem of calculus twice, we obtain
	\begin{equation}
	\begin{aligned}
	&\ \CE_{\text{av}}(t)(\rho)-\CU(t,0)(\rho)\\
	=&\ \big(\CE_{\text{av},1}(t)(\rho)-\CE_{\text{av},0}(t)(\rho)\big)-\big(\CU_1(t,0)(\rho)-\CU_0(t,0)(\rho)\big)\\
	=&\ \int_{0}^{1}\mathrm{d}u\int_{0}^{u}\mathrm{d}v\ \frac{\mathrm{d}^2}{\mathrm{d}v^2}\big[\CE_{\text{av},v}(t)(\rho)-\CU_v(t,0)(\rho)\big] \\
	=&\ \int_{0}^{1}\mathrm{d}u\int_{0}^{u}\mathrm{d}v\ 
	\bigg\{e^{-iv \int_{0}^{t}\mathrm{d}\tau H(\tau)}
	\bigg[-i\int_{0}^{t}\mathrm{d}\tau H(\tau),\bigg[-i\int_{0}^{t}\mathrm{d}\tau H(\tau),\rho\bigg]\bigg]
	e^{iv\int_{0}^{t}\mathrm{d}\tau H(\tau)}\\
	&\ \qquad\qquad\qquad-\int_0^t\mathrm{d}\tau\ p(\tau)e^{-iv \frac{H(\tau)}{p(\tau)}}
	\bigg[-i\frac{H(\tau)}{p(\tau)},\bigg[-i\frac{H(\tau)}{p(\tau)},\rho\bigg]\bigg]
	e^{iv \frac{H(\tau)}{p(\tau)}}\bigg\}.
	\end{aligned}
	\end{equation}
	
	We take $\sigma$ to be a state on the joint system of the original register and an ancilla register with the same dimension. Using properties of the Schatten norms, we have
	\begin{equation}
	\begin{aligned}
	&\norm{\big(\CE_{\text{av}}(t)\otimes\mathds{1}\big)(\sigma)-\big(\CU(t)\otimes\mathds{1}\big)(\sigma)}_1\\
	\leq&\int_{0}^{1}\mathrm{d}u\int_{0}^{u}\mathrm{d}v\ 
	\bigg\{\norm{\bigg[-i\int_{0}^{t}\mathrm{d}\tau H(\tau)\otimes\mathds{1},\bigg[-i\int_{0}^{t}\mathrm{d}\tau H(\tau)\otimes\mathds{1},\sigma\bigg]\bigg]}_1\\
	&\qquad\qquad\quad+\int_0^t\mathrm{d}\tau\ p(\tau)
	\norm{\bigg[-i\frac{H(\tau)}{p(\tau)}\otimes\mathds{1},\bigg[-i\frac{H(\tau)}{p(\tau)}\otimes\mathds{1},\sigma\bigg]\bigg]}_1\bigg\} \\
	\leq&\int_{0}^{1}\mathrm{d}u\int_{0}^{u}\mathrm{d}v\ \bigg[4\norm{H}_{\infty,1}^2+4\int_0^t\mathrm{d}\tau\ \frac{\norm{H(\tau)}_\infty^2}{p(\tau)}\bigg].
	\end{aligned}
	\end{equation}
	Using the definition $p(\tau)=\norm{H(\tau)}_{\infty}/\norm{H}_{\infty,1}$, the second term of the integrand can be further simplified as
	\begin{equation}
	\begin{aligned}
	\int_0^t\mathrm{d}\tau\ \frac{\norm{H(\tau)}_\infty^2}{p(\tau)}
	=\norm{H}_{\infty,1}^2,
	\end{aligned}
	\end{equation}
	giving
	\begin{equation}
	\norm{\big(\CE_{\text{av}}(t)\otimes\mathds{1}\big)(\sigma)-\big(\CU(t,0)\otimes\mathds{1}\big)(\sigma)}_1
	\leq 4\norm{H}_{\infty,1}^2.
	\end{equation}
	Optimizing over $\sigma$ proves the claimed bound.
\end{proof}

\end{document}